\documentclass[12pt]{article}
\usepackage{mheck}

\usepackage{longtable}

\usepackage{blkarray}

\usepackage{tikz} 
\usetikzlibrary{matrix}

\theoremstyle{theorem}
\newtheorem*{thm}{Theorem}

\newtheorem*{fact}{Fact}

\newtheorem{pro}{Proposition}

\newtheorem{lem}{Lemma}[section]
\newtheorem*{cla}{Claim}

\newtheorem*{nec}{Necessary condition}

\theoremstyle{definition}
\newtheorem{defn}{Definition}
\newtheorem{rem}{Remark}

\newtheorem{exe}{Example}

\def\Dsl{\,\raise.15ex\hbox{/}\mkern-13.5mu D}

\def\BQ{\mathbb{Q}}

\date{March, 2018}

\title{Special Arithmetic of Flavor
}

\authors{Matteo Caorsi\footnote{e-mail: {\tt matteocao@gmail.com}} and Sergio Cecotti\footnote{e-mail: {\tt cecotti@sissa.it}}\vskip 9pt

\centerline{SISSA, via Bonomea 265, I-34100 Trieste, ITALY}}

\abstract{We revisit the classification of rank-1 4d $\cn=2$ QFTs
in the spirit of Diophantine Geometry, viewing their special geometries as elliptic curves over the chiral ring (a Dedekind domain). The Kodaira-N\'eron model maps the space of non-trivial rank-1 special geometries to the well-known moduli of pairs $(\ce,F_\infty)$ where $\ce$ is a relatively minimal, rational elliptic surface with section, and $F_\infty$ a fiber with additive reduction. Requiring enough Seiberg-Witten differentials yields a condition on $(\ce,F_\infty)$ equivalent to the ``safely irrelevant conjecture''. The Mordell-Weil group of $\ce$ (with the N\'eron-Tate pairing) contains a canonical root system arising from $(-1)$-curves in special position in the N\'eron-Severi group. This canonical system is identified with the roots of the flavor group $\mathsf{F}$: 
the allowed flavor groups are then read from the Oguiso-Shioda table of Mordell-Weil groups. Discrete gaugings correspond to base changes. Our results are consistent with previous work by Argyres \emph{et al.}}

\begin{document}

\maketitle

\newpage

\tableofcontents

\newpage

\section{Introduction}

The classification of all 4d $\cn=2$ SCFTs of rank $k$ may be (essentially) reduced to the geometric problem of classifying all dimension $k$ special geometries \cite{SW0,SW1,Donagi:1995cf,donagi,ACSW,ACS2,Argyres:2015ffa,Argyres:2015gha}. This classification is naturally organized in two distinct steps. At the \emph{coarse-grained} level one lists the allowed $k$-tuples $\{\Delta_1,\Delta_2,\cdots,\Delta_k\}$ of dimensions  of operators generating the Coulomb branch 
(see refs.\!\cite{Argyres:2018zay,Caorsi:2018zsq,Argyres:2018urp}
   for recent progress on this problem). Then we have the \emph{fine} classification of the physically inequivalent models belonging to each  coarse-grained class, that is, the list of the distinct QFT  which share the same dimension $k$-tuple. Theories in the same coarse-grained class differ by invariants like the flavor symmetry group, the conformal charges $k_F$, $a$, $c$, and possibly by subtler aspects. 
For $k=1$ the fine classification has been worked out by Argyres \emph{et al.} in a  series of remarkable papers \cite{Argyres:2015ffa,Argyres:2015gha,Argyres:2016xua,Argyres:2016xmc,Argyres:2016yzz}. 
  To restrict the possibilities, these authors invoke some physically motivated conjectures like ``planarity'', 
``absence of dangerous irrelevant operators'', and ``charge quantization''.
\medskip

The purpose of this note is to revisit the fine classification for $k=1$, introducing new ideas and techniques which we hope may be of help for a future extension of the fine classification beyond $k=1$. In the process we shall greatly simplify and clarify several points of the $k=1$ case and provide proofs of (versions of) the above conjectures. 
\medskip

We borrow the main ideas from Diophantine Geometry\footnote{\ For a survey see \cite{diop}.}. The present paper is meant to be a first application of the arithmetic approach to Special Geometry which we dub \textit{Special Arithmetic}. We hope the reader will share our opinion that Special Arithmetic is a very beautiful and deep way of thinking about Special Geometry.
\medskip

Traditionally, Special Geometry is studied through its Weierstrass model. In this note we  advocate instead  the use of the Kodaira-N\'eron model, which we find both easier and more powerful.  The Kodaira-N\'eron model $\ce$ of the (total) space $X$ of a non-trivial\footnote{\ \emph{Non-trivial} means that $X$ is not the product of an open curve $C$ with a fixed elliptic curve $E$ (equivalently: $X$ has at least one singular fiber); physically,  \emph{non-trivial} means the 4d $\cn=2$ theory is not free.} rank-1 special geometry is a (smooth compact) relatively minimal, elliptic surface\footnote{\ A complex surface is said to be \emph{elliptic} if it has a holomorphic fibration over a curve, $\ce\to C$, whose \emph{generic} fiber is an elliptic curve.}, with a zero section $S_0$,  which happens to be rational (so isomorphic to $\mathbb{P}^2$ blown-up at 9 points). $\ce$ is equipped with a marked fiber $F_\infty$  which must be \emph{un}stable\footnote{\ \textit{Un}stable fibers are also known as \emph{additive} fibers. In this paper we shall use mostly the latter name.}, that is, as a curve $F_\infty$ is \emph{not} semi-stable in the Mumford sense. Comparing with Kodaira classification, we get 11 possible $F_\infty$: seven of them correspond to the (non-free) Coulomb branch dimensions $\Delta$ allowed in a rank-1 SCFT, and the last four to the possible \emph{non}-zero values of the $\beta$-function in a rank-1 asymptotically-free $\cn=2$ theory.
The fact that $\ce$ is rational implies \emph{inter alia} the ``planarity conjecture'', that is, the chiral ring $\mathscr{R}$ is guaranteed to be a polynomial ring (of transcedence degree 1),  $\mathscr{R}=\C[u]$.

Basic arithmetic gagdets  associated to $\ce$ are its Mordell-Weil group $\mathsf{MW}(\ce)$ of ``rational'' sections,
its finite-index sub-group
$\mathsf{MW}(\ce)^0$ of ``narrow''  sections, and its finite sub-set of ``integral'' sections, all sections being exceptional $(-1)$-curves on $\ce$. $\mathsf{MW}(\ce)^0$ is a  finitely-generated free Abelian group i.e.\! a lattice. This lattice is naturally endowed with a positive-definite, symmetric, integral pairing 
\be
\langle-,-\rangle_{NT}\colon \mathsf{MW}(\ce)^0\times \mathsf{MW}(\ce)^0\to\Z
\ee
 induced by the N\'eron-Tate (canonical) height. The root system $\Xi_\infty$ of the flavor group $\mathsf{F}$ may be identified with a certain finite sub-set of $\mathsf{MW}(\ce)^0$, and is completely determined by the above arithmetic data (we shall give a sketchy picture of $\Xi_\infty$ momentarily). Given this identification, the list of possible flavor symmetries in rank-1 $\cn=2$ QFT is read directly from the well-known tables of Mordell-Weil  groups for rational elliptic surfaces, see ref.\!\cite{oguiso} or the nice book \cite{bookMW}.

The result are (of course) consistent with
Argyres \emph{et al.}  \cite{Argyres:2015ffa,Argyres:2015gha,Argyres:2016xua,Argyres:2016xmc,Argyres:2016yzz}.

The Diophantine language is useful for other questions besides classifying the flavor groups.
First it clarifies the subtler distinctions between inequivalent geometries which have the same flavor group and invariants $k_F$, $a$, $c$. Second, from the arithmetic point of view the gauging of a discrete symmetry is a \emph{base change} (an extension of the ground field $K$ over which the elliptic curve is defined). All consistent base changes are listed in table 6 of \cite{cover}; from that table one recovers  the discrete gauging classification.

\subsection{A sketch of flavor symmetry in rank-1 4d $\cn=2$ QFT}

A rational $(-2)$-curve $\cc$ on a rational elliptic surface $\ce$ is called an \textit{$E_8$-root curve} if it is disjoint from the zero section $S_0$, i.e.\! iff it satisfies the three conditions
\be\label{ppppmz}
\cc\cdot S_0=0,\qquad \cc^2=-2,\qquad K_\ce\cdot\cc=0.
\ee

The name ``$E_8$-root curve'' stems from the following fact.
Consider the ``most generic'' situation\footnote{\ More precisely, the ``most generic'' \emph{unstable} elliptic surface $\ce$.} where all fibers of the elliptic surface
are \emph{irreducible} curves. Physically, such a geometry describes a general mass-deformation of the Minahan-Nemeshanski (MN)
SCFT \cite{Minahan:1996cj} with flavor symmetry $\mathsf{F}=E_8$. On such a surface, $\ce_{M\!N8}$, there are precisely 240 $E_8$-root curves $\cc_a$ ($a=1,\dots,240$) 
 with the property that their intersection pairing
 \be
 \cc_a\cdot\cc_b=-\big\langle \alpha_a,\,\alpha_b\big\rangle_\text{Cartan},\qquad a,b=1,\cdots,240,
 \ee
where the $\alpha_a$ are the roots of $E_8$ and $\langle-,-\rangle_\text{Cartan}$ is the bilinear form induced by the $E_8$ Cartan matrix. In other words, the classes of these 240 rational curves form an $E_8$ root system in\footnote{\ $H_2(\ce,\R)^\perp$ denotes the subspace of classes orthogonal to the fiber and the zero section $S_0$.} $H_2(\ce,\R)^\perp$. This is the usual way we understand geometrically the presence of a flavor $E_8$ symmetry in this particular Minahan-Nemeshanski theory.

The crucial observation is that the $E_8$-root $(-2)$-curves are in one-to-one correspondence with a special class of exceptional $(-1)$-curves on $\ce$, namely the finite set of \emph{integral} elements of the Mordell-Weil group $\mathsf{MW}(\ce)$. Indeed, $\cc$ is an  $E_8$-root curve 
if and only if the $(-1)$-curve
($\equiv$ section of $\ce$)
\be\label{kkkkkq1z}
S=\cc+S_0+F
\ee
is \emph{integral} in $\mathsf{MW}(\ce)$. 
In \eqref{kkkkkq1z} $S_0$ and $F=-K_\ce$ are the divisors of the zero section and a fiber, respectively, $=$ being equality in the N\'eron-Severi (or Picard) group.
\smallskip 

Away from this ``generic'' situation, three competing mechanisms become operative:
\begin{description}
\item[Symmetry lift.] Part of the original $E_8$ root system gets lost. An elliptic surface $\ce$ with a reducible fiber has less than 240 integral sections and hence less than 240 $E_8$-root curves satisfying \eqref{ppppmz}. Some of the $E_8$ roots simply are no longer there. For instance, the elliptic surface $\ce_{M\!N7}$ describing a generic mass deformation of the $E_7$ Minahan-Nemeshanski SCFT
 has only 126+56=182 integral sections and hence only 182 $E_8$-root curves; 
\item[Symmetry obstruction.]
Some of the $E_8$-root curves present in $\ce$
do not correspond to symmetries because 
they are obstructed by the symplectic structure $\Omega$ of special geometry. An $E_8$-root curve $\cc$ leads to a root of the flavor symmetry $\mathsf{F}$ if and only if, for all irreducible components of the fibers $F_{u,\alpha}$,
\be\label{jjjzz6321}
F_{u,\alpha}\cdot \cc=0.
\ee
If our rank-1 model is not $E_8$ MN, some fiber of $\ce$ should be reducible, and hence we get extra  conditions \eqref{jjjzz6321} from the additional irreducible components. These conditions project some $E_8$-curves  out of the root system of $\mathsf{F}$. For instance, 56 of the 182 $E_8$-root curves of $E_7$ MN do not satisfy \eqref{jjjzz6321}, and we remain with only 126 ``good'' curves making the $E_7$ root system.
The other 56 $E_8$-root curves on $\ce_{M\!N7}$ yield instead the weights of the fundamental representation $\textbf{56}$ of $E_7$.\\
If a $(-2)$-curve $\cc$ with the properties \eqref{ppppmz} satisfies \eqref{jjjzz6321} we say that it lays \textit{in good position} in the N\'eron-Severi lattice. Only $E_8$-curves in good position contribute to the flavor symmetry. The corresponding sections \eqref{kkkkkq1z} are precisely the ones which are both integral and narrow; 
\item[Symmetry enhancement.]
Some integral sections which are not of the form \eqref{kkkkkq1z} -- and hence not related to the ``generic'' symmetry -- but lay in good position in the N\'eron-Severi group, get promoted to roots of the flavor Lie algebra $\mathfrak{f}=\mathfrak{Lie}(\mathsf{F})$. When both kinds of roots are present -- the ones inherited from the ``original'' $E_8$ as well as the ones arising from enhancement -- the last (first) set makes the long (short) roots of a non-simply-laced Lie algebra. \end{description}
 
In some special models the symmetry enhancement has a simple physical meaning. These $\cn=2$ QFTs may be  obtained by gauging a discrete (cyclic) symmetry of a parent theory. This situation is described geometrically by a branched cover between the corresponding elliptic surfaces
\be\label{kkocer}
f\colon \ce_\text{parent}\to\ce_\text{gauged}.
\ee
 The $(-1)$-curves associated with the enhanced symmetries of the gauged QFT, when pulled back to the parent ungauged geometry $\ce_\text{parent}$, take the form \eqref{kkkkkq1z} for some honest $E_8$-root curve $\cc_a\subset \ce_\text{parent}$ laying in good position. Thus, at the level of the parent theory the ``enhanced'' symmetries are just the ``obvious'' flavor symmetries inherited from $E_8$-roots. The deck group of \eqref{kkocer}, $\mathsf{Gal}(f)$ (the symmetry being gauged), is a subgroup of $\mathrm{Aut}(\ce_\text{parent})$ and acts on the parent root system by isometries of its lattice. The root system of $\ce_\text{gauged}$ is obtained by ``folding'' the Dynkin graph of $\ce_\text{parent}$ by its symmetry $\mathsf{Gal}(f)$.  In these particular examples the general arithmetic construction of the flavor symmetry $\mathsf{F}$ out of the Mordell-Weil group $\mathsf{MW}(\ce_\text{gauged})$ is equivalent to the physical relation between the flavor symmetries of the gauged and ungauged QFTs.
 \smallskip
 
\noindent For examples of flavor root lattices, see \S.\ref{lllazzzq}. For examples of diagram foldings see \S.\ref{exxplexe}.

\paragraph{Organization of the paper.} The rest of this paper is organized as follows. In section 2 we show the relation between rank-1 special geometries and rational elliptic surfaces. We also discuss UV completeness in this context. In section 3 we introduce the notion of ``SW completeness'', that is, the requirement of having sufficiently many SW differentials, and show that this condition has the same consequences as the ``safely irrelevant conjecture''. Then we review the Mordell-Weil groups, and show how they can be used to determine the allowed flavors groups. In section 4 we discuss base change, and present further evidence of the general picture.

\section{Special geometry and rational elliptic surfaces}

\subsection{Preliminaries}

In the literature there are several ``morally equivalent'' definitions of ``Special Geometry''. In order not to confuse the reader, we state explicitly the definitions we use in this paper.
\medskip

We start from the most basic, physically defined, object: the \emph{chiral ring} $\mathscr{R}$, i.e.\! the ring of all (quantum) chiral operators in the given 4d $\cn=2$ theory. $\mathscr{R}$ is a commutative associative $\C$-algebra with unit and also a finitely-generated domain.

\begin{rem} \emph{A priori}  we do not require   $\mathscr{R}$ to be a free polynomial ring; this fact  will be proven below (in the case of  interest). Neither we assume $\mathscr{R}$ to be  normal, i.e.\! in principle we allow for the ``exotic'' possibilities discussed in ref.\!\cite{Argyres:2017tmj}, but rule  them out (in rank 1) as a result of the analysis.
\end{rem}  

\begin{defn} The \textit{Coulomb branch} $M$ of $\mathscr{R}$ is the complex-analytic variety $M$ underlying the affine scheme $\mathsf{Spec}\,\mathscr{R}$. Its complex dimension is called the \emph{rank} of $\mathscr{R}$. We write $\C(M)$ for the function field of $M$ i.e.\! the field of fractions of the domain $\mathscr{R}$.
\end{defn}

\begin{rem} In rank-1, the normalization $\mathscr{R}^\text{nor}$ of the chiral ring $\mathscr{R}$ is a Dedekind domain, so morally ``\,$\mathscr{R}^\text{nor}$ behaves like the ring of integers $\Z$\,''. This is the underlying reason why classification in rank-1 is so simple.
\end{rem}

\begin{defn} Let $\mathscr{R}$ be a finitely-generated domain over $\C$ of dimension $k$. A \textit{special geometry} (SG) \text{over $\mathsf{Spec}\,\mathscr{R}$} is a quadruple $(\mathscr{R},X,\Omega,\pi)$ where: 
\begin{itemize}
\item[a)] $X$ is a complex space of dimension $2k$ and $\Omega$ a holomorphic symplectic form on $X$;
\item[b)] $\pi\colon X\to M$ is a holomorphic fibration, with base the Coulomb branch $M$ of $\mathscr{R}$, such that the fibers $F_u\equiv \pi^{-1}(u)$ are Lagrangian, i.e.\! $\Omega|_{F_u}=0$ for $u\in M$; 
\item[c)] $\pi$ has a (preferred) section $s_0\colon M\to X$. We write $S_0:=s_0(M)$
for its image;
\item[d)] the fiber $F_\eta$ over the generic point $\eta$ of $M$ is (isomorphic to) a polarized Abelian variety. The restriction $S_0|_{F_\eta}$  is the zero in the corresponding group.
\end{itemize}
\end{defn}

In other words, a special geometry is a (polarized) Abelian variety over the function field $\C(M)$
which, as a variety over $\C$, happens to be symplectic  with Lagrangian fibers.

\begin{rem} The Coulomb branch $M$ is an open space, so the definition of special geometry should be supplemented by appropriate ``boundary conditions'' at infinity. Physically, the  requirement is that the geometry should be asymptotic to the UV behavior of either a unitary SCFT or an asymptotically free QFT. In the context of rank-1 special geometries, this condition (dubbed \emph{UV completeness}) will be made mathematically precise in \S.\ref{uvcccom}.  
\end{rem}

\begin{defn} A \textit{Seiberg-Witten (SW) differential} $\lambda$ on a special geometry is a meromorphic one-form $\lambda$ on $X$ such that $d\lambda=\Omega$. We are only interested in special geometries admitting SW differentials. We shall say that a special geometry is \emph{SW complete} iff it admits ``enough'' SW differentials, that is,
all infinitesimal deformations of the symplectic structure $\Omega$ may be induced by infinitesimal deformations of $\lambda$ and viceversa.
\end{defn}

\subsection{Rank-1 special geometries as rational elliptic surfaces} 

\subparagraph{Kodaira-N\'eron models.} Let $\mathscr{R}$ be a rank-1 chiral ring and $\eta\in\mathsf{Spec}\,\mathscr{R}$  the generic point of its Coulomb branch. The fiber over $\eta$, $F_\eta$, is open and dense in $X$, and may be identified with its ``good'' locus of smooth fibers.  In rank-1, $F_\eta$ is (in particular) an elliptic curve $E(\C(M))$ defined over the function field $\C(M)$ of transcendence degree 1.
By a \emph{model} of the elliptic curve $E(\C(M))$ we mean a morphism $\pi\colon \ce\to C$ between an algebraic surface $\ce$ and a curve $C$ whose generic fiber is isomorphic to the elliptic curve $E(\C(M))$ (i.e.\! to $F_\eta$).  All models are birationally equivalent, and contain the same amount of information. Most of the literature on Special Geometry uses the minimal Weierstrass model, $y^2=x^3+a x+b$, ($a,b\in \C(M)$), which is easy to understand but has the drawback that it is not smooth (in general) as a complex surface. A better tool is the Kodaira-N\'eron model given by a (relatively minimal\footnote{\ An elliptic surface is \emph{relatively minimal} if its fibers do not contain exceptional $-1$ rational curves.}) smooth compact surface $\ce$ fibered over a smooth compact curve $C$ such that $\C(M)\cong \C(C)$. The Kodaira-N\'eron model always exists for one-dimensional function fields \cite{kod1,kod2,ell2,silver2,bookMW}, and is unique up to isomorphism. In particular, the smooth model exists for all rank-1 special geometries.
By definition, the generic fiber of $\pi\colon\ce\to C$ is a smooth elliptic curve, and $\ce$ is a smooth, relatively minimal, (compact) elliptic surface having a section. The geometry of such surfaces is pretty well understood, see e.g.\! \cite{kod1,kod2,silver2,miranda,bhpvv,bookMW}. 
Note that, having a section, the surface $\ce$ cannot have  multiple fibers.\smallskip

We say that an elliptic surface is \emph{trivial} iff $\ce\cong E\times C$, that is,
iff its fibers are all smooth elliptic curves. This trivial geometry corresponds to a free $\cn=2$ QFT. We shall esclude the trivial case from now on, that is, for the rest of the paper we assume that at least one\footnote{\ If $\ce$ has singular fibers, it has at least 2 of them. If it has precisely 2 singular fibers, its functional invariant is constant, and the special geometry describes an interacting SCFT with no mass deformation.} fiber of $\ce$ is singular.
Special geometries with this property will be called \emph{non-free}.
In the non-free case \cite{kod2,miranda,bhpvv,bookMW},
\be
q(\ce)\equiv h^{0,1}(\ce)=0.
\ee
The non-smooth fibers which may appear in $\ce$
are the ones in the Kodaira list, see table \ref{kodairatypes}.

\begin{rem}(\textit{Weierstrass vs.\! Kodaira-N\'eron}) The (minimal) Weierstrass model is obtained from the smooth Kodaira-N\'eron surface, $\ce$, by blowing-down all components of the \emph{reducible} fibers which do not cross the reference section $S_0$. If all exceptional fibers are irreducible (i.e.\! of Kodaira types $I_1$ and $II$) the two models coincide, and the flavor group is the ``generic'' $E_8$. Otherwise the blowing-down introduces singularities in the Weierstrass geometry. From the Weierstrass viewpoint, the information on the flavor group is contained in these singularities, which are most easily analyzed by blowing-up them. By construction, this means working with the Kodaira-N\'eron model.  
\end{rem}

\subparagraph{The chiral ring $\mathscr{R}$ is free.}
Since $\C(C)\cong \C(M)$, we have 
\be
M=C\setminus \mathsf{supp}\,D_\infty
\ee
 for some effective divisor $D_\infty$. Then
\be\label{kkkasqw}
X\cong \ce\setminus  \mathsf{supp}\,\pi^*(D_\infty).
\ee
In order to be a special geometry, $X$ must be symplectic (the fibers of $\pi$ are then automatically Lagrangian). 
From eqn.\eqref{kkkasqw} we have
\be\label{xxx123z}
\Omega\in \Gamma(\ce,K_\ce(\pi^*D_\infty)),
\ee
so that\footnote{\ Here $\sim$ denotes linear equivalence.} $K_\ce(\pi^*D_\infty)\sim \co_\ce$, or 
\be\label{jjjaqw1}
K_\ce=-\pi^*[D_\infty].
\ee
We recall Kodaira's formula for the canonical divisor $K_\ce$ of an elliptic surface with no multiple fibers (see e.g.\! \S.V.12 of \cite{bhpvv})
\be\label{jjjaqw2}
K_\ce =\pi^*\,\cl, \quad \text{where }\cl\ \text{is a line bundle on $C$ of degree }p_g(\ce)+2g(C)-1,
\ee
where $p_g(\ce)\equiv h^{2,0}(\ce)\geq0$ (resp.\! $g(C)\equiv h^{1,0}(C)\geq 0$) is the geometric genus of $\ce$ (resp.\! $C$). Comparing eqns.\eqref{jjjaqw1}\eqref{jjjaqw2}, yields $\cl\sim -D_\infty$;
since $D_\infty$ is effective, $\deg \cl<0$,  which is consistent with eqn.\eqref{jjjaqw2} only if
\be\label{numnum1}
p_g(\ce)=g(C)=0\quad\Longrightarrow\quad C\cong \mathbb{P}^1\ \text{and}\ \deg D_\infty=1,
\ee
so that $D_\infty$ consists of a single point on $\mathbb{P}^1$ which we denote as $\infty$. The Coulomb branch is
\be
M=\mathbb{P}^1\setminus \infty=\C,
\ee
and its ring of regular functions is $\mathscr{R}\cong \C[u]$.

\subparagraph{The functional invariant $\mathscr{J}$.} The elliptic fibration $\pi\colon \ce\to\mathbb{P}^1$ yields a rational function (called the \textit{functional invariant} of $\ce$ \cite{kod1,kod2})
\be
\mathscr{J}\colon \mathbb{P}^1\to\mathbb{P}^1,
\qquad u\mapsto J(\tau_u),
\ee
where (for $u\in\mathbb{P}^1\equiv \pi(\ce)$) $\tau_u\in\mathfrak{H}$ is the modulus of the elliptic curve $\pi^{-1}(u)$ and $J(z)\equiv j(z)/1728$, $j(z)$ being the usual modular invariant \cite{ser}. The function $\mathscr{J}$ determines $\ce$ up to quadratic transformations \cite{miranda}. A quadratic transformation consists in flipping the type of an \emph{even} number of fibers according to the rule
\be\label{starruke}
I_b\leftrightarrow I_b^*,\quad II\leftrightarrow IV^*,\quad III\leftrightarrow III^*,\quad IV\leftrightarrow II^*.
\ee

\subparagraph{Scale-invariant vs.\! mass-deformed special geometries.} As we shall see momentarily, the special geometries associated to scale-invariant $\cn=2$ SCFT are precisely the ones described by a \emph{constant} function $\mathscr{J}$. Mass-deformed geometries instead have functional invariants of positive degree, $\deg\mathscr{J}>0$. Our approach applies uniformly to both situations.

\subparagraph{The surface $\ce$ is rational.} The divisor $-K_\ce$ is effective, so all plurigenera vanish (i.e.\! $\ce$ has Kodaira dimension $\kappa(\ce)=-\infty$). Since $q(\ce)=0$, $\ce$ is rational by the Castelnuovo criterion \cite{bhpvv}. 
The other numerical invariants of $\ce$ are \cite{miranda,bhpvv}:
\be\label{numnum2}
\text{topological Euler number $e(\ce)=12$,}\quad b_1(\ce)=0,\quad b_2(\ce)=h^{1,1}(\ce)=\varrho(\ce)=10.
\ee
The N\'eron-Severi group $\mathsf{NS}(\ce)\cong \mathsf{Pic}(\ce)/\mathsf{Pic}(\ce)^0$ is then a unimodular (odd) lattice of signature $(1,9)$. 
In facts $\ce$, being a relatively minimal rational elliptic surface with section, is just $\mathbb{P}^2$ blown-up at 9 points (see \textbf{Theorem 5.6.1} of \cite{enriques} or \S.\,VIII.1 of \cite{miranda}). Note that the Kodaira-N\'eron surfaces of all rank-1 special geometries have the same topological type allowing for a uniform discussion of them. This does not hold in the Weierstrass approach, since the blowing down kills cohomology classes in a model dependent fashion.   
\medskip

The class $F$ of any fiber is $-K_\ce$. By the moving lemma $F^2=0$, so $K^2_\ce=0$. Let $S_0$ be the zero section. One has $K_\ce\cdot S_0=-F\cdot S_0=-1$. Then, by adjunction,
\be
-2=2\,g(\mathbb{P}^1)-2= S_0^2+K_\cs\cdot S_0\quad\Longrightarrow \quad S_0^2=-1,
\ee 
so that the zero section $S_0$ is an exceptional $(-1)$-line. Contracting it we get a weak degree 1 del Pezzo surface, see \S.\,8.8.3 of \cite{dolga}.

The Euler number of $\ce$ is the sum of the 
Euler numbers of its singular fibers. So
\be\label{jjasqwj}
12=e(\ce)=\sum_{u\in U}e(F_u),
\ee
where $U\subset\mathbb{P}^1$ is the finite set of points with a non-smooth fiber.
The Euler numbers $e(F)$ for the various types of singular fibers are listed in table \ref{kodairatypes}. Note that for all $\text{additive}^\ast$ fibers
$e(F^*)\geq 6$, so eqn.\eqref{jjasqwj} implies that we can have at most one $\text{additive}^\ast$ fiber with the single exception of $\{I_0^*,I_0^*\}$ which is (the Kodaira-N\'eron model of) the special geometry of $\cn=4$ SYM with gauge group $SU(2)$. Since a quadratic transformation preserves the parity of the number of $\ast$, the function $\mathscr{J}$ specifies completely $\ce$ if there are no $\text{additive}^\ast$ fibers, while if there is one such fiber we are free to flip the type of the $\text{additive}^\ast$ fiber and of precisely one other fiber (possibly regular) by the rule \ref{starruke}.
This process is called \emph{transfer of $\ast$} \cite{mira}.

\subparagraph{$\ce$ and the symplectic structure of $(X,\Omega)$.} 
We write $F_\infty=\pi^{-1}(\infty)$ for the fiber at infinity. Then
\be
X=\ce\setminus F_\infty.
\ee
From eqn.\eqref{xxx123z} we see that the
pair $(\ce,F_\infty)$ uniquely fixes the symplectic structure $\Omega$ up to overall normalization.
Physically, the overall constant may be seen as a choice of mass unit.

\begin{table}
\begin{center}
\begin{tabular}{c|c|cccccc|c|c}\hline\hline
category & type & $e(F)$ & $m(F)$ & $m(F)^{(1)}$ & $o(F)$ & $u(F)$ & $d(F)$ & \begin{scriptsize}semi-simple\end{scriptsize} & $R(F)$ \\\hline
stable (regular) & $I_0$ & 0 & 1 &1& 0&0 &0 &\checkmark\\\hline
semi-stable & $I_{b\geq1}$ & $b$ & $b$& $b$ &-&-& $b$ &no& $A_{b-1}$\\\hline
$\text{additive}^\circ$ & $II$ & 2 & $1$ &1 &1&0& 0&\checkmark\\
& $III$ & 3 & $2$&2&0&1& 0&\checkmark& $A_1$\\
& $IV$ & 4 & $3$&3&2&0&0&\checkmark &$A_2$\\\hline
$\text{additive}^*$ & $I_{0}^*$ & $6$ & $5$&4&0&0& 0& \checkmark& $D_4$\\ 
& $I_{b\geq1}^*$ & $6+b$ & $5+b$&4 &-&-& $b$ & no& $D_{4+b}$ \\
&$II^*$ & 10 & $9$&1&2&0&0& \checkmark& $E_8$\\
& $III^*$ & 9 & $8$&2&0&1&0&\checkmark&$E_7$\\
& $IV^*$ & 8 & $7$&3&1&0&0&\checkmark& $E_6$ \\\hline\hline
\end{tabular}
\caption{Kodaira fibers and their numerical invariants. $I_0$ is the \textit{regular} (generic) fiber, all other types are \textit{singular.} Additive fibers are also called \emph{unstable}.  Additive fibers come in two categories: \textit{un-starred} and \textit{starred} ones. A fiber is simply-connected iff it is additive; then $e(F)=m(F)+1$. A fiber type is \emph{reducible} if it has more than one component, i.e.\! $m(F)>1$. A fiber $F$ is \emph{semi-simple} iff the local monodromy at $F$ is semi-simple. The last column yields the intersection matrix of the non-identity component of the reducible fibers. $m(F)^{(1)}$ is the number of simple components in the divisor $F_u$ equal to the order of the center of the simply-connected Lie group in the last column.}
 \label{kodairatypes}
\end{center}
\end{table}

\subparagraph{Moduli of rational elliptic surfaces with given singular fibers.} The rational elliptic surfaces with a given set of singular fiber types, $\{F_u\}_{u\in U}$, are in one-to-one correspondence with the rational functions $\mathscr{J}$ consistent with the given fiber types $\{F_u\}_{u\in U}$ modulo the action of $\mathrm{Aut}(\mathbb{P}^1)\equiv PSL(2,\C)$. We adopt the convention that the number of fibers of a given Kodaira type is denoted by the corresponding lower-case roman numeral, so (say) $iii$ stands for the number of fibers of type $III$ while $iv^*$ for the number of fibers of type $IV^*$.  We also write $s$, $a^\circ$, and $a^*$ for, respectively, the total number of semi-stable, $\text{additive}^\circ$, and $\text{additive}^\ast$ singular fibers (cfr.\! table \ref{kodairatypes}). 

As already anticipated, we distinguish two kinds of geometries: 
\begin{description}
\item[scale-invariant:]
$\mathscr{J}$ is constant, that is, the coupling $\tau_u$ does not depend on the point $u$ in the Coulomb branch $M$,
and the special geometry is scale-invariant. The fiber configurations of the elliptic surfaces with $\mathscr{J}$ constant which satisfy the physical requirement of UV completeness (see \S.\ref{uvcccom}) are listed in table \ref{constmi}. Each of the first three elliptic surfaces describe \emph{two} distinct $\cn=2$ SCFT, having Coulomb branch dimension $\Delta$ and $\Delta/(\Delta-1)$, depending on which of the two singular fibers is placed at $\infty$: see \S.\ref{uvcccom}, in particular eqn.\eqref{idedeime};
\item[mass-deformed:] $\mathscr{J}$ has positive degree $d>0$ and satisfies the following properties \cite{kod2,miranda} (cfr.\! table \ref{kodairatypes}):
\begin{itemize}
\item $\mathscr{J}$ has a pole of order $b$ at fibers of types $I_b$, $I^*_b$;
\item the order of zero $\nu_0(F)$ of $\mathscr{J}$  at a fiber of type $F$ is $\nu_0(F)=o(F)$ mod 3;
\item the order of zero $\nu_1(F)$ of $\mathscr{J}-1$  at a fiber of type $F$ is $\nu_1(F)=u(F)$ mod 2.
\end{itemize}
\end{description}
Since the degree of $\mathscr{J}$ is $d=\sum_u d(F_u)$, $\mathscr{J}$ is constant iff all fibers are semi-simple.

From table \ref{kodairatypes} we see that 
\be\label{erelation}
e(F)=d(F)+2o(F)+3u(F)+\begin{cases} 0 & F\in\{\text{semi-stable}\}\cup\{\text{additive}^\circ\}\\
6 & F\in\{\text{additive}^\ast\}.
\end{cases}
\ee
Suppose $d>0$, and let $a_i$ be the positions of the poles, $b_j$ the positions of 0's and $c_k$ the positions of 1's of $\mathscr{J}$. We have
\be\label{lllazq9}
\mathscr{J}\!(z)= A\frac{\prod_{j=1}^{n_0}(z-b_j)^{\nu_0(j)}}{\prod_{i=1}^{p}(z-a_i)^{b(i)}}=1+B\frac{\prod_{k=1}^{n_1}(z-c_k)^{\nu_1(k)}}{\prod_{i=1}^{p}(z-a_i)^{b(i)}},
\ee
where $n_0$ (resp.\! $n_1$) is the (maximal)  number of \emph{distinct} 0's (resp. 1's) 
\be\label{lllazq8}
n_0=ii+iv+ii^*+iv^*+\frac{d-\sum_i o(F_i)}{3},\qquad 
n_0=iii+iii^*+\frac{d-\sum_i u(F_i)}{2}, 
\ee
and $p$ is the number of non-semi-simple fibers, $p\equiv s+\sum_{b\geq1}i^*_b$.
$PSL(2,\C)$ allows to fix three points; the number of effective parameters is then  $n_0+n_1+p-1$, while the equality of the two expressions in \eqref{lllazq9} yields $d+1$
relations. Thus the space of rational functions has dimension $\mu\equiv n_0+n_1+p-d-2$,
or
\be
\mu+i^*_0=s+a^\circ+2a^*-\frac{1}{6}\Big[d+2\sum_i o(F_i)+3\sum_i u(F_i)+6a^*\Big]-2
\ee 
 The number of fibers of a given type is restricted by eqn.\eqref{jjasqwj}. Using \eqref{erelation}
 \be
 12=d+\sum_i(2o(F_i)+3 u(F_i))+6 a^*,
 \ee
so that
\be
\mu+i^*_0=s+a^\circ+2a^*-4.
\ee
Hurwitz formula applied to the covering
$\mathscr{J}\colon \mathbb{P}^1\to \mathbb{P}^1$ implies \cite{miranda}
\be\label{bunddd}
\mu+i^*_0\geq 0.
\ee 
A fiber configuration $\{F_u\}_{u\in U}$ which violates the bound \eqref{bunddd} cannot be realized geometrically.
The bound is saturated if and only if: \textit{i)} $\mathscr{J}$ is a Belyi function\footnote{\ Recall that a function $C\to \mathbb{P}^1$ ($C$ a compact Riemann surface) is a (normalized)  \emph{Belyi function} iff it ramifies only over the three points $\{0,1,\infty\}$.} \cite{abc,bookgraph} and \textit{ii)}
the order of the zeros of $\mathscr{J}$ (resp.\! of 
$\mathscr{J}-1$) is $\leq3$ (resp.\! $\leq 2$) \cite{miranda}. 

\begin{table}
$$
\begin{tabular}{cccc}\hline\hline
$\phantom{\Big|}\{II^*,II\}$ & $\{III^*,III\}$& $\{IV^*,IV\}$ & $\{I_0^*,I_0^*\}\phantom{\Big|}$\\\hline\hline
\end{tabular}
$$
\caption{List of singular fiber configurations for $\mathscr{J}$ \emph{constant}
containing at most one $\text{additive}^\circ$ fiber
\cite{mira}. They describe SCFTs with masses and relevant perturbations switched off.}\label{constmi}
\end{table}

The number of parameters from which a $d>0$ special geometry $(X,\Omega)$ depends  is
 \be\label{kkaswo}
 \boldsymbol{n}\equiv \mu+i^*_0+1\equiv s+a^\circ+2a^* -3,
 \ee
 where the term $i_0^*$ arises from the choice of the locations where we insert the $I_0^*$ fibers (by quadratic transformation of some regular fiber $I_0$) and the $+1$ is the overall scale of $\Omega$.
 
\subparagraph{$ADE$ and all that.} The exceptional fibers $F_u$
 are in general reducible with $m(F_u)$
 irreducible components $F_{u,\alpha}$, see table \ref{kodairatypes}. The divisor of $\pi^{-1}(u)$ has the form
 \be
 \big(\pi^{-1}(u)\big)= \sum_{\alpha=0}^{m(F_u)-1} n_\alpha\, F_{u,\alpha},
 \ee
  where the $n_\alpha$ are positive integers. A component $F_{u,\alpha}$ is said to be \emph{simple} iff $n_\alpha=1$. The numbers of simple components for each fiber type,  $m(F)^{(1)}$, are listed in  table \ref{kodairatypes}. By the moving lemma
 we have 
 \be
 F\sim \sum_{\alpha=0}^{m(F_u)-1} n_\alpha\, F_{u,\alpha} \quad\text{and}\quad F_{u,\alpha}\cdot F=0\ \ \text{for all }u,\alpha.
 \ee
 
Let $S_0$ be the zero section. Since $F_u\cdot S_0=1$ for all $u$, the section $S_0$ intersects a single component of the fiber $F_u$ which must be simple. This component is said to be the \emph{identity component,} and will be  denoted as $F_{u,0}$.
 Forgetting the identity component $F_{u,0}$, we remain with
the set $F_{u,\alpha}$, $\alpha=1,\cdots,m(F_u)-1$ of irreducible divisors whose intersection matrix
 \be
 F_{u,\alpha}\cdot F_{u,\beta}=-C(F_u)_{\alpha\beta}\qquad \alpha,\beta=1,\cdots, m(F_u)-1
 \ee 
 is minus the Cartan matrix $C(F_u)$ of the $ADE$ root system $R(F_u)$. The root systems $R(F)$ for the various fiber types are listed in the last column of table \ref{kodairatypes}. One has
 \be\label{rankEu}
 \mathrm{rank}\,R(F)=e(F)-\begin{cases}
 1 & F\in \text{semi-stable}\\
 2 & F\in \text{additive}^\circ\cup\text{additive}^\ast.\end{cases}
 \ee
 To each $R(F)$ we associate a finite Abelian group
 \be\label{kkkaszqg}
 Z(F)=\Gamma_{\text{weigth}\atop\text{lattice}}\Big/\Gamma_{\text{root}\atop\text{lattice}}
 \ee
 isomorphic to the center of the simply-connected Lie group associated to $R(F)$. From the table
 we see that $|Z(F)|=m(F)^{(1)}$, and indeed, $Z(F)$ acts freely and transitively on the simple components of a reducible fiber. See table \ref{zkoda}.
 
 \begin{table}
 \begin{center}
 \begin{tabular}{c||c|c|c|c|c|c|c}\hline\hline
 $F$ & $I_{b<2}$ & $I_{b\geq2}$ & $I^*_b\;$ $b$ even & $I^*_b$ $b\;$ odd & $II$, $II^*$ & $III$, $III^*$ & $IV$, $IV^*$\\\hline
 $Z(F)$ & $\{0\}$ & $\Z/b\Z$ & $\Z/2\Z\times \Z/2\Z$ & $\Z/4\Z$ & $\{0\}$ & $\Z/2\Z$ & $\Z/3\Z$\\\hline\hline
 \end{tabular}
 \caption{The Abelian group $Z(F)$ of a Kodaira fiber of type $F$.}\label{zkoda}
 \end{center}
 \end{table}

  \subparagraph{Allowed fiber configurations and Dynkin theorem.} A fundamental problem is
  to list the configurations of singular fibers, $\{F_u\}_{u\in U}$ which are realized by some rational elliptic surface. There are 379 fiber
  configurations which satisfy eqn.\eqref{jjasqwj}.
  Of these 100 cannot be geometrically realized, most of them because they violate the Hurwitz bound \eqref{bunddd}. For the list of those which \emph{can} be realized see refs.\!\cite{per,mira}. 
  
  The realizable fiber configurations may be understood in Lie-theoretic terms. From its numerical invariants, eqns.\eqref{numnum1}\eqref{numnum2}, we infer that $\ce$, seen as a compact topological 4-fold, has intersection form $H_2(\ce,\Z)\times H_2(\ce,\Z)\to\Z$ isomorphic to $\boldsymbol{U}\oplus E_8^-$,
  where $E_8^-$ stands for the $E_8$ root lattice with the opposite quadratic form (see \S.\ref{digresssion} for details). The classes of the non-identity components of the reducible fibers belong to the $E_8^-$ part, so that homology yields an embedding of roots lattices \cite{bookMW,per,mira,oguiso}
  \be
  \bigoplus_{\text{reducible}\atop\text{fibers }F_u} R(F_u)\rightarrowtail E_8.
  \ee
  Two such embeddings are equivalent if they are conjugate by the Weyl group $\mathrm{Weyl}(E_8)$. The classification of all inequivalent embeddings was given by Dynkin \cite{dyn}. There are 70 root systems which may be embedded in $E_8$, all but 5 of them in an unique way. The special 5 have two inequivalent embeddings each. They are
\be\label{speccase}
A_7,\quad A_3\oplus A_3,\quad A_5\oplus A_1,\quad A_3\oplus A_1\oplus A_1,\quad A_1\oplus A_1\oplus A_1\oplus A_1.
\ee  
Three out of the 70 sub-root systems cannot be realized geometrically because they violate Euler's bound \eqref{jjasqwj}. The full list of allowed singular fiber configurations, $\{F_u\}_{u\in U}$, is then obtained by consider the various ways of producing a given allowed embedding of a root system in $E_8$.

 \subparagraph{Aside: \emph{Dessin d'enfants}.}
 
 When the bound \eqref{bunddd} is saturated, 
 the functional invariant $\mathscr{J}$ is (in particular) a Belyi function.
 Belyi functions are encoded in their Grothendieck  \emph{dessin d'enfants} \cite{abc,bookgraph}. Since it is often easier to work with \emph{dessins} than with  functions, we recall that story even if we don't need it.\footnote{\ For a survey  see \cite{abc}.}
A Belyi function $f$ is a holomorphic map from some Riemann surface $\Sigma$ to $\mathbb{P}^1$ which is branched only over the three points $0$, $1$ and $\infty$. If a Belyi functions exists, $\Sigma$ and $f$ are defined over the a number field. The \emph{dessin} of $f$ is a graph $G\subset\Sigma$ which is the inverse image of the segment $[0,1]\subset\mathbb{P}^1$. The inverse images of 0 (resp.\! 1) are represented by white\footnote{\ We use the coloring convention of \cite{abc}. Ref.\cite{bookgraph} uses the opposite convention.} nodes $\circ$ (black nodes $\bullet$). The coloring makes $G$ into a bi-partite graph.
$G$ is a connected graph whose complement,
$\Sigma\setminus G$ is a disjoint union of disks in one-to-one correspondence with the inverse images of $\infty$.
 
If the bound  \eqref{bunddd} is saturated, all
white (black) nodes have valency at most 3 (2).

\begin{exe} The \emph{dessins} of Argyres-Douglas of type $A_2$ and of pure $SU(2)$ SYM are (the first one is drawn in a chart of $\mathbb{P}^1$ around $\infty$)
\be
\begin{aligned}
&\text{Argyres-Douglas $A_2$:} &&\begin{gathered}
\xymatrix{\circ\ar@{-}@/^1pc/[rr]\ar@{-}@/_1pc/[rr]&&\bullet}
\end{gathered}\\
&\phantom{---}\\
&\text{pure SYM:}
&&\begin{gathered}
\xymatrix{\bullet \ar@{-}@/^1pc/[rr] 
\ar@{-}@/_1pc/[rr]&& \circ 
\ar@{-}[r]& \bullet\ar@{-}[r]&\circ\ar@{-}@/^1pc/[rr] 
\ar@{-}@/_1pc/[rr]&&\bullet}
\end{gathered}
\end{aligned}
\ee
These are special instances of \emph{double flower dessins} \cite{bookgraph} so that the special geometry for these QFTs is \emph{rational} (i.e.\! defined over $\BQ$).
\end{exe}

 If
the bound \eqref{bunddd} is not saturated, so that the space $\cs(\{F_u\})$ of rational functions $\mathscr{J}$ has positive dimension, and $\deg\mathscr{J}\geq 2$, we may still found some \emph{exceptional} points $\cp_\sigma\subset 
\cs(\{F_u\})$ where $\mathscr{J}$ becomes a Belyi function (however the nodes will have larger valency).

\begin{exe} Consider the fiber configuration $\{II; I_4,I_1^6\}$ which corresponds to the Argyres-Wittig SCFT \cite{ARWi} with $\Delta=6$ and flavor symmetry $Sp(10)$. It has $\mu=4$, that is, $\boldsymbol{n}=5\equiv \mathrm{rank}\,\mathfrak{sp}(10)$. The bound \eqref{bunddd} is far from being saturated, but nevertheless there is a dimension 1 locus in the space of mass parameters where the model is described by the \emph{dessin}\vglue10pt
\be
\xymatrix{\bullet\ar@{-}@/^1.8pc/[rr]\ar@{-}@/_1.8pc/[rr]\ar@{-}[rr] &&\circ\ar@{-}@/^1.8pc/[rr]\ar@{-}@/_1.8pc/[rr]\ar@{-}[rr]  &&\bullet\ar@{-}[rr]&& \circ\ar@{-}@/^1.8pc/[rr]\ar@{-}@/_1.8pc/[rr]\ar@{-}[rr]  && \bullet }
\ee
\end{exe}
\vglue24pt

\subsection{UV completeness and the fiber $F_\infty$ at infinity}\label{uvcccom}

As already mentioned, the possible fibers $F_\infty$ at $\infty$ are rescricted by the condition of ``UV completeness''. 
Heuristically this means that we can make sense out  of the QFT without introducing extra degrees of freedom at infinite energy (they would play the role of Pauli-Vilards regulators that we cannot get rid off). This translates in the condition that $F_\infty$ is simply-connected, hence \emph{additive} ($\equiv$ unstable).
There are only 11 additive fibers which can appear in a rational elliptic surface
\be\label{inftytypes}
\overbrace{II,\ III,\ IV,\ II^*,\ III^*,\ IV^*,\ I_0^*}^\text{semisimple},\ \overbrace{I^*_1,I^*_2,I^*_3,I^*_4}^\text{non-semisimple}.
\ee
$I^*_{b\geq 7}$ are ruled out because their Euler number $>12$ and the fiber configurations $\{I_6^*\}$, $\{I_5^*,I_1\}$ because they have
$\mu=-2$ and $-1$ respectively.
The seven semi-simple fibers in \eqref{inftytypes} correspond to the seven UV asymptotic special geometries\footnote{\ By ``UV asymptotic special geometry'' we mean the behavior of the geometry for large $u\in\C$.} for a non-free SCFT, which are labelled by the dimension $\Delta$ of the chiral operator parametrizing the Coulomb branch:
\be\label{idedeime}
\begin{tabular}{c|ccccccc}
additive semi-simple $F_\infty$ & $II$ & $III$ & $IV$ & $II^*$ & $III^*$ & $IV^*$ & $I_0^*$\\\hline
$\Delta$ & 6 & 4 & 3 & $6/5$ & $4/3$ & $3/2$ & 2
\end{tabular}
\ee
 while the 4 non semi-simple ones describe the possible UV behavior of asymptotically-free theories. Note that the correspondence between fiber type at infinity, $F_\infty$, and the Coulomb branch dimension, $\Delta$, is the \emph{opposite} of the usual one since the monodromy at infinity $M_\infty$ in the Coulomb branch  is related to the local monodromy around the fiber at infinity, $M(F_\infty)$, by an inversion of orientation  
\be
M_\infty = M(F_\infty)^{-1}.
\ee 
This is consistent with the usual statements in the SCFT context, since in the zero-mass limit $\ce$ becomes
a constant geometry with fiber configuration $\{F_\infty;F_0\}$ with 
\be
M_\infty=M(F_0)\equiv M(F_\infty)^{-1},
\ee
 and in the literature it is usually stated the  zero-mass limiting correspondence $F_0\leftrightarrow \Delta$.

\subparagraph{Asymptotically free QFTs.} $F_\infty =I^*_b$ yields the UV asymptotic special geometry  of $SU(2)$ SYM coupled to $N_f=4-b$ fundamentals. This relation implies both the UV geometrical bound $b\leq 4$ and the physical UV bound $N_f\leq 4$, and illustrates as the additive reduction of the fiber at infinity captures the physical idea of UV completeness (i.e.\! $\beta\leq0$).

$SU(2)$ with $N_f$ fundamentals and generic masses corresponds to the fiber configuration $\{I_{4-N_f}^*;I_1^{N_f+2}\}$. Using eqn.\eqref{kkaswo} we see that the number of parameters on which this geometry depends is
\be
\boldsymbol{n}(N_f)=N_f+1
\ee
which is the physically correct number: 
the masses and the Yang-Mills scale $\Lambda$ for $N_f\leq 3$, the masses and the coupling constant $g_\text{YM}$ for $N_f=4$ (which correspond to $+i_0^*$ in eqn.\eqref{kkaswo}).

$\{I_4^*;I_1^2\}$ is the only fiber configuration with $F_\infty=I_4^*$ \cite{per}; it corresponds to an \emph{extremal} rational elliptic surface \cite{mirper} (defined over $\BQ$). Thus pure $SU(2)$ SYM is unique in its UV class. 
There are two configurations with $F_\infty=I_3^*$, $\{I_3^*;I_1^3\}$ and $\{I_3^*,II,I_1\}$; the second one will be ruled out in \S.\ref{kkkaqwhh} on the base that is has no ``enough'' SW differentials. Hence $SU(2)$ SQCD with $N_f=1$ is also unique in its UV class. There are six configurations with
$F_\infty=I_2^*$, three of which are ruled out by the same argument. The remaining 3 are either the standard SQCD or special cases of it.
Finally, there are 13 configurations with $F_\infty=I^*_1$; 8 of them are ruled out as before, while 5 look like special instances of SQCD with $N_f=3$.

\subparagraph{The UV asymptotics of the
special geometry.}
The behavior of the periods $(b(u),a(u))$  as we approach $u=\infty$ for each of the 11 allowed fibers at infinity, eqn.\eqref{inftytypes},
may be read (including the sub-leading corrections!) in table (VI.4.2) of \cite{miranda}. If $u$ is a standard coordinate on the Coulomb branch, as $u\to\infty$ the special geometry periods behave as
\be\label{jjjzzzh}
\big(b(u),a(u)\big)= \big(u\,r_2(1/u), u\,r_1(1/u)\big)\qquad u\ \text{large,}
\ee
where the functions $r_1(t)$, $r_2(t)$ are listed in the table of ref.\!\cite{miranda}. In the particular case of a geometry which is UV asymptotic to a SCFT, $F_\infty$ is semi-simple, and $a(u)\simeq u^{1/\Delta}$ with $\Delta$ as in eqn.\eqref{idedeime},
 confirming the correspondence $F_\infty\leftrightarrow \Delta$.

\subparagraph{The ``generic'' massive deformation.} As an example, let us consider the generic configuration
with a marked fiber $F_\infty$ of one type in  eqn.\eqref{inftytypes}, i.e.\!
$\{F_\infty; I_1^{12-e(F_\infty)}\}$, which is always geometrically realized. The number of parameters $\boldsymbol{n}(F_\infty)$ in the geometry is
\be
\boldsymbol{n}(F_\infty)=12-e(F_\infty)-\begin{cases}2 & F_\infty\in \{\text{additive}^\circ\}\\
1 &F_\infty\in \{\text{additive}^\ast\},
\end{cases} 
\ee
 which precisely matches the number of physical relevant+marginal deformations for the theory with Coulomb dimension $\Delta$ having the largest possible flavor symmetry of rank
 \be
\mathrm{rank}\,\mathfrak{f}= 8-\mathrm{rank}\,R(F_\infty)\equiv 10-e(F_\infty).
 \ee

\section{SW differentials vs.\! Mordell-Weil lattices}

We have not yet enforced one crucial property of the special geometries relevant for $\cn=2$ QFT, namely the existence of Seiberg-Witten (SW) differentials with the appropriate properties. In this section we consider the restrictions on the pair $(\ce,F_\infty)$ coming from this requirement. 

\subsection{SW differentials and horizontal divisors}

 A SW differential $\lambda$ is a meromorphic one-form on the total space $X=\ce\setminus F_\infty$ or,
with non-trivial residue along a  simple normal-crossing effective divisor $D_{SW}$,
such that $d\lambda=\Omega$ in $X$. Let $D_{SW}=\sum_i D_i$, be the decomposition of $D_{SW}$ into prime divisors.  Standard residue formulae \cite{log1,log2,log3} yield the following equality in cohomology \cite{SW1}
 (see \cite{donagi} for a nice discussion in the present context) 
\be\label{kkkaqweX}
[\Omega]=\sum_i \mu_i[D_i],
\ee
where the complex coefficients $\mu_i$
 are linearly related to the masses $m_a$ living in the
Cartan subalgebra $\mathfrak{h}$ of the flavor Lie algebra $\mathfrak{f}=\mathfrak{Lie}(\mathsf{F})$ \cite{SW0,SW1}. For the relation of this statement to the Duistermaat-Heckman theorem in symplectic geometry, see \cite{donagi}. We may rewrite \eqref{kkkaqweX} in terms of the independent mass parameters $m_a$ as
 \be\label{llllasqw}
 [\Omega]=\sum_{a=1}^{\mathrm{rank}(\mathfrak{f})} m_a[L_a],
 \ee
 for certain \emph{non} effective divisors $L_a$ on $X$. The surface $\ce$ (with a choice of zero section $S_0$) has an involution corresponding to taking the negative in the associated Abelian group.
Since $\lambda$ is odd under this involution, the divisors $L_a$ belong to the odd cohomology \cite{SW1}.
 
The closure in the smooth elliptic surface $\ce$  of the divisors $D_i$, $L_a$  (originally defined in the open quasi-projective variety $X\subset\ce$) yields divisors on $\ce$ which we denote by the same symbols. 

A divisor on an elliptic surface $\pi\colon\ce\to M$ contained (resp.\! not contained) in a fiber is called \emph{vertical} (resp.\! \emph{horizontal}) \cite{miranda,bookMW}. 
The divisors $D_i$, $L_a$ cannot be contained in a fiber
$F$ of $\ce$, since the masses are well-defined  at all generic points $u\in M$ and $u$ independent.\footnote{\ A more formal argument is as follows. The primitive divisors contained in the fibers are compact analytic submanifolds of $X$, hence as cohomology classes have type $(1,1)$ while $\Omega$ has type $(2,0)$.} We conclude that the divisors $D_i$, $L_a$ are horizontal.  
 Since the fibers are Lagrangian and the $m_a$ independent, eqn.\eqref{llllasqw} implies\footnote{\ Again, this also follows from type considerations.}
 \be\label{hhhas}
 \Omega\big|_{F_{u,\alpha}}=0\quad\Longrightarrow\quad F_{u,\alpha}\cdot L_a=0\ \ \text{for all }a,\; u,\;\alpha.
 \ee
 
 Thus, to determine the flavor symmetry $\mathsf{F}$ associated to a given special geometry $(\ce,F_\infty)$, preliminarly  we have to  understand the geometry of its horizontal divisors.
In the next subsection we  review this elegant topic. We shall resume the discussion of Special Geometry in \S.\,\ref{swari}.

 \subsection{Review: N\'eron-Severi and Mordell-Weil groups}\label{digresssion}

\subparagraph{The N\'eron-Severi group.}
 We see the divisors $D_i$, $L_a$ on $\ce$ as elements of the Ner\'on-Severi group $\mathsf{NS}(\ce)$,
  the group of divisors on $\ce$ modulo algebraic equivalence. For all projective variety $Y$,
the N\'eron-Severi group $\mathsf{NS}(Y)$ is a finitely-generated Abelian group \cite{log1,bookMW}. Its rank,
$\varrho(Y):=\mathrm{rank}\,\mathsf{NS}(Y)$, is called the \emph{Picard number} of $Y$.

In the case of a projective surface $S$, the intersection pairing $\langle-,-\rangle$ endows\footnote{\ The free Abelian group $\mathsf{Num}(S)$ is the group of divisors modulo \emph{numerical} equivalence. } \be\mathsf{Num}(S):=\mathsf{NS}(S)/\mathsf{NS}(S)_\text{tors}\ee
 with a non-degenerate, symmetric, integral, bilinear pairing of signature
 $(1,\varrho(S)-1)$ having the same parity as the first Chern class. In other words, $(\mathsf{Num}(S), \langle-,-\rangle)$ is
 a non-degenerate \emph{lattice}.

 For an elliptic surface $\ce$, the N\'eron-Severi group is torsion-free, so $\mathsf{Num}(\ce)=\mathsf{NS}(\ce)$, and the N\'eron-Severi group is itself a lattice. 
 
 If, in addition, the elliptic surface $\ce$ is \emph{rational}, we have the further identification with the Picard group: 
 $\mathsf{Num}(\ce)=\mathsf{NS}(\ce)=\mathsf{Pic}(\ce)$, that is, linear and numerical equivalence coincide. In this case $p_g(\ce)=0$, $\varrho(\ce)=10$, and
 $\mathsf{NS}(\ce)$ is an (odd) unimodular lattice
 of signature $(1,9)$; by general theory it is isomorphic to
 \be\label{nnnssssa}
\boldsymbol{U}\oplus E_8^-,
 \ee
 where $\boldsymbol{U}$ is the rank 2 lattice with Gram matrix
 \be
  \begin{bmatrix} -1 & 1\\ 1 & 0\end{bmatrix},
  \ee
 and  $E_8^-$ is the opposite\footnote{\ Given a lattice $L$, by its \emph{opposite} lattice $L^-$ we mean the same Abelian group endowed with a bilinear pairing which is \emph{minus} the original one.} of the $E_8$ root lattice (its Gram matrix is \emph{minus} the Cartan matrix of $E_8$). $E^-_8$ is the \emph{unique} negative-definite, even, self-dual lattice of rank 8 \cite{ser}.
 The sublattice $\boldsymbol{U}$ in \eqref{nnnssssa} is spanned by the zero section $S_0$ and the fiber $F$.
 
The N\'eron-Severi group $\mathsf{NS}(\ce)$ of a rational elliptic surface contains an obvious subgroup, called the \emph{trivial group,} $\mathsf{Triv}(\ce)$,
generated by the zero section $S_0$ and all the vertical divisors, that is, the irreducible divisors $F_{u,\alpha}$ contained in some fiber $F_u$. The rank of the trivial group is
\be\label{kkkas1228z}
\mathrm{rank}\,\mathsf{Triv}(\ce)=2+\sum_{u\in U} \big(m(F_u)-1\big)\geq 2
\ee
 where $U\subset\mathbb{P}^1$ is the finite set of points at which the fiber is not smooth and
 $m(F_u)$ is the number of irreducible components $F_{u,\alpha}$ of the fiber
 at $u$ (see table \ref{kodairatypes}).
 The only relations between the vertical divisors $F_{u,\alpha}$ are $\sum_\alpha n_\alpha\, F_{u,\alpha}=F$, from which we easily get eqn.\eqref{kkkas1228z}.  In facts, $\mathsf{Triv}(\ce)$ is the lattice
  \be\label{nnnssssaII}
\boldsymbol{U}\oplus R^-,
 \ee
where $R^-$ is the lattice generated by all irreducible components of the fibers which do not meet the zero section $S_0$. 
As reviewed in the previous section, the opposite lattice $R$ of $R^-$ is the direct (i.e.\! orthogonal) sum of the roots lattices of  $ADE$ type associated to each reducible fiber
(see last column of table \ref{kodairatypes})
\be\label{nnnssssa2}
R=\bigoplus_{u\colon m(F_u)>1} R(F_u).
\ee

\begin{defn}\label{essential}
The orthogonal complement $R^\perp$ of $R$ in $E_8$ is called the \emph{essential} lattice $\Lambda$.
\end{defn}

\subparagraph{The group $Z(\ce)$.}
The intersection form $\langle-,-\rangle$ yields a map
\be
\mathsf{NS}(\ce)\to R^\vee:=\mathrm{Hom}(R,\Z),
\ee
and, passing to the quotient, (cfr.\! eqn.\eqref{kkkaszqg})
\be
\gamma\colon \mathsf{NS}(\ce)\to R^\vee/R\equiv \bigoplus_{u\colon m(F_u)>1} Z(F_u)=: Z(\ce).
\ee
We note that
\be\label{rrtyu}
\mathsf{Triv}(\ce)\subset \mathrm{ker}\,\gamma.
\ee

\subparagraph{$E_8$-root curves.} A rational curve $\cc\subset\ce$ is said to be a \emph{$E_8$-root curve}
iff its class $\cc\in \mathsf{NS}(\ce)$ is a root of the $E_8^-$ lattice (cfr.\! eqn.\eqref{nnnssssa}).
In other words, $\cc$ is a $E_8$-root curve iff
the following three conditions are satisfied
\be\label{3condd}
F\cdot\cc\equiv -K_\ce\cdot\cc=0,\qquad S_0\cdot \cc=0,\qquad
\cc^2=-2.
\ee 
An \emph{$E_8$-root curve} is a particular case of a $(-2)$-curve \cite{dolga}. It is clear that a rational elliptic surface $\ce$ may have at most 240 $E_8$-root curves (240 being the number of roots of $E_8$).

\subparagraph{The Mordell-Weil group of sections.} 
As discussed in section 2,
a rank-1 special geometry is, in particular,
an elliptic curve $E/K$ defined over the field of rational functions $K\equiv \C(u)$. The Mordell-Weil group $\mathsf{MW}(E/K)$ of an elliptic curve $E$ defined over some field $K$ is the group $E(K)$ of its points which are ``rational'' over $K$, that is, whose coordinates lay in $K$ and not in some proper field extension  \cite{ell1,ell2,ell3,diop}.
When $K$ is a number field, the Mordell-Weil theorem of Diophantine Geometry states\footnote{\ The Mordell-Weil and the N\'eron-Lang theorems are stated in general for arbitrary Abelian varieties, not just for elliptic curves.} that the Abelian group $E(K)$ is finitely-generated \cite{ell1,ell2,ell3,diop}. When $K$ (as in our case) is a function field defined over $\C$, the Mordell-Weil theorem must be  replaced by the N\'eron-Lang one \cite{diop,NL}: there is
an Abelian variety $B$ over $\C$ of dimension $\leq 1$
(an Abelian variety of dimension zero being just the trivial group $0$), and an injective map defined
over $K$ \cite{traaaa} 
\be
\mathrm{tr}_{K/\C}\colon B\to E,\qquad \text{(the \emph{trace} map)}
\ee
such that the quotient group $E(K)/\mathrm{tr}_{K/\C}(B)$ is finitely generated.   

We may rephrase the above Diophantine statements in geometric language in terms of our Kodaira-N\'eron model, which is a rational elliptic 
surface $\pi\colon\ce\to\mathbb{P}^1$ with a reference section $s_0\colon \mathbb{P}^1\to \ce$. The (scheme-theoretic) closure in $\ce$ of a point of $E$ defined over $\C(u)$ is the same as a section of $\pi$.
Thus the set of all sections of $\pi$
is an Abelian group (with respect to fiberwise addition) isomorphic to the ``abstract'' Mordell-Weil group $\mathsf{MW}(\ce)\equiv \mathsf{MW}(E/K)$. The preferred section $S_0$ (the image of $s_0$) plays the role of zero in this group.

The Abelian variety $B/\C$ is non-trivial iff
the fibers $F_u$ of $\ce$
are all isomorphic elliptic curves; in this case
$\ce\cong B\times\mathbb{P}^1$ and the special geometry is \emph{trivial}. 
As before, we focus on non-trivial geometries where $B=0$. Then the group $\mathsf{MW}(\ce)$
is finitely generated by the N\'eron-Lang theorem.

A section $S$ defines a horizontal divisor on $\ce$. By Abel theorem, addition in $\mathsf{MW}(\ce)$ corresponds to addition in $\mathsf{NS}(\ce)/\mathsf{Triv}(\ce)\equiv\mathsf{Pic}(\ce)/\text{(vertical classes)}$
\be
S_1+S_2=S_3\ \text{in }\mathsf{MW}(\ce)\quad\Longleftrightarrow\quad (S_1)+(S_2)=(S_3)\ \text{in }\mathsf{NS}(\ce)/\mathsf{Triv}(\ce),
\ee
so that, in our special case, the N\'eron-Lang theorem follows from the finite-generation of the N\'eron-Severi group.

The basic result is

\begin{thm}[Thm.\,(VII.2.1) of \cite{miranda}, Thm.\,6.5 of \cite{bookMW}] Let $\ce$ be a (relatively minimal) rational elliptic surface. The following sequence (of finitely-generated Abelian groups) is exact
\be\label{rqwert}
0\to \mathsf{Triv}(\ce)\to\mathsf{NS}(\ce)\xrightarrow{\;\beta\;}\mathsf{MW}(\ce)\to 0.
\ee
In particular, the Shioda-Tate formula holds
\be
\mathsf{rank}\,\mathsf{MW}(\ce)=8-\sum_{u\in U}\big(m(F_u)-1\big).
\ee
In addition, using \eqref{rrtyu}, the map $\gamma$ factors through $\mathsf{MW}(\ce)$ so we get a map
\be\label{rrrvzx}
\gamma\colon \mathsf{MW}(\ce)\to Z(\ce)
\ee
which is injective on the torsion subgroup.
\end{thm}

\begin{rem} The involution of $\ce$ acts on the Abelian group $\mathsf{MW}(\ce)$
as $S\mapsto -S$. Hence the even cohomology is in the kernel of $\beta$.
\end{rem}

From eqns.\eqref{nnnssssa}\eqref{nnnssssaII}\eqref{nnnssssa2} we see that (after flipping the overall sign!!)
\be
\mathsf{MW}(\ce)\cong  E_8\Big/R\equiv E_8\left/\bigoplus_{m(F_u)>1} R(F_u)\right..
\ee

The exact sequence \eqref{rqwert} does not split (in general). However it does split once tensored with $\mathbb{Q}$. Then we define $\mathsf{NS}(\ce)_\BQ:=\mathsf{NS}(\ce)\otimes \mathbb{Q}$.
The orthogonal projection
\be\label{spplet0}
\Phi_\BQ\colon \mathsf{MW}(\ce)\to \mathsf{NS}(\ce)_\BQ,
\ee
 splits $\beta$. Explicitly \cite{bookMW},
\be\label{spplet}
\Phi_\BQ\colon  S\longmapsto S-S_0+(1+S\cdot S_0)F+\sum_{u\in U}\sum_{\alpha,\beta=1}^{m(F_u)-1} F_{u,\alpha}\, C(u)^{-1}_{\alpha\beta} (F_{u,\alpha}\cdot S)\in\mathsf{NS}(\ce)_\BQ,
\ee
 whose image (by construction) is contained in the essential subspace (cfr.\! \textbf{Definition \ref{essential}}) 
 \be
 \Lambda_\BQ\subset E_8\otimes \BQ\subset \mathsf{NS}(\ce)_\BQ.
 \ee
  In eqn.\eqref{spplet}
 $C(u)_{\alpha\beta}$ is the Cartan matrix of the $ADE$ root system $R(F_u)$, cfr.\! table \ref{kodairatypes}.
 
The map $\Phi_\BQ$ induces on $\mathsf{MW}(\ce)/\mathsf{MW}(\ce)_\text{tors}$ a $\BQ$-valued positive-definite symmetric pairing, called the \emph{N\'eron-Tate pairing}
\be
\langle S_1,S_2\rangle_\text{NT} =\langle \Phi_\BQ(S_1),\Phi_\BQ(S_2)\rangle_{E_8\otimes\BQ}\in\BQ.
\ee
The corresponding quadratic form $S\mapsto \mathsf{h}(S)\equiv \langle S,S\rangle_\text{NT}$ is known as the N\'eron-Tate (or canonical) \emph{height}.
In terms of the intersection pairing $\cdot$ we have \cite{bookMW}
\be\label{kkkasqwe}
\langle S_1,S_2\rangle_\text{NT}= 1+ S_1\cdot S_0+S_2\cdot S_0-S_1\cdot S_2 -\sum_{m(F_u)>1} C(u)^{-1}_{\alpha\beta}\,(F_{u\alpha}\cdot S_1)(F_{u\beta}\cdot S_2)\in\frac{1}{m}\Z
\ee
where $m=\mathrm{lcm}(m(F_u)^{(1)})$).
$\mathsf{MW}(\ce)/\mathsf{MW}(\ce)_\text{tors}$ equipped with the N\'eron-Tate pairing is called the \textit{Mordell-Weil lattice} \cite{bookMW}.

\begin{rem} From the facts that $K_\ce=-F$ and
$S^2=-\chi(\ce)=-1$, we see that all sections $S$ are, in particular, (rational) $(-1)$-curves.
\end{rem}

\subparagraph{The narrow Mordell-Weil group.}

There is an important finite-index torsion-free subgroup of
$\mathsf{MW}(\ce)$, the \emph{narrow}
Mordell-Weil group, $\mathsf{MW}(\ce)^0$, consisting of the sections which at all reducible  fibers intersect the same component $F_{u,0}$ as $S_0$,
so that the sum in the \textsc{rhs} of eqn.\eqref{spplet} vanishes. The sum in eqn.\eqref{kkkasqwe} also vanishes if either $S_1$ or $S_2$ is narrow. Thus the N\'eron-Tate pairing is $\Z$-valued when restricted to $\mathsf{MW}(\ce)^0$.
More generally, the pairing of a section in $\mathsf{MW}(\ce)^0$ with \emph{any} section in $\mathsf{MW}(\ce)$ is an integer.
Indeed one has the isomorphisms of lattices \cite{bookMW}
\be
\mathsf{MW}(\ce)^0\cong\Lambda,\qquad
\mathsf{MW}(\ce)/\mathsf{MW}(\ce)_\text{tors}\cong \Lambda^\vee.
\ee
One shows that $\mathsf{MW}(\ce)^0=\mathrm{ker}\,\gamma$ \cite{miranda}.

\subparagraph{Integral sections.} Given a (fixed) particular model of an elliptic curve $E/\Bbbk$ over a number field $\Bbbk$, say an explicit curve in $\mathbb{A}_\Bbbk^2$, we may consider, besides the points which are ``rational'' over $\Bbbk$, also the points which are ``integral'' over $\Bbbk$, that is, whose coordinates belong to the Dedekind domain $\mathfrak{O}_\Bbbk$ of algebraic integers in $\Bbbk$. While the ``rational'' points of $E/\Bbbk$ form a (typically infinite) finitely-generated group, its ``integral'' points form a \emph{finite set} (Siegel theorem \cite{ell3}). 

The integer ring $\mathfrak{O}_{\C(u)}$ of
the rational function field $\C(u)$ is, of course, the Dedekind domain of polynomials in $u$, $\C[u]$. The analogy with Siegel theorem in Number Theory suggests to look for sections given by polynomials. Of course, ``integrality'' is a model-dependent statement.  
If we focus on the elliptic curves over the rational field $\C(u)$ which are relevant for Special Geometry, and describe them through their \emph{minimal} Weierstrass model,
$y^2=x^3+a(u)x+b(u)$, the correct statement is that the integral sections are the ones 
of the form $(x,y)=(p(t),q(t))$ where $p(t)$ (resp.\! $q(t)$) is a polynomial of degree at most 2 (resp.\! 3)
\cite{bookMW}. 

From the vantage point of the Kodaira-N\'eron model the notion of integral section becomes simpler:

\begin{defn}
A section $S\in\mathsf{MW}(\ce)$ is said to be \emph{integral} if it does not intersect the zero section, i.e.\! $S\cdot S_0=0$.
\end{defn}

Siegel theorem still holds \cite{bookMW}:
\begin{pro} $\ce$ a (relatively minimal) rational elliptic surface.  There are only finitely many integral sections (at most 240) and they generate the full Mordell-Weil group. 
\end{pro}

Indeed, from eqn.\eqref{kkkasqwe} we see that
if $S$ is integral
\be\label{kkkazzqqav}
\mathsf{h}(S)=\langle S, S\rangle_\text{NT}=2-\sum_{u,\alpha} C(F_u)^{-1}_{\alpha,\alpha}\;(F_{u,\alpha}\cdot S), 
\leq 2\ee
so that all integral sections have square-norms $\leq 2$. Since there are only finitely many such elements in the lattice $\Lambda^\vee$
and the torsion subgroup $\subseteq \Lambda^\vee/\Lambda$ is finite, the statement follows.

\begin{lem} \label{kkkz1i}
If $S\in\mathsf{MW}(\ce)$ satisfies any two of the following three conditions, it also satisfies the third one:
\begin{itemize}
\item[1)] $S$ is narrow: $S\in\mathsf{MW}(\ce)^0$;
\item[2)] $S$ is integral: $S\cdot S_0=0$;
\item[3)] $S$ has N\'eron-Tate height 2: $\mathsf{h}(S)=2$.
\end{itemize}
\end{lem}

\begin{proof}
From eqn.\eqref{kkkasqwe}, the narrow condition implies
$\mathsf{h}(S) = 2+2\,S\cdot S_0\geq 2$
with equality if and only if $S\cdot S_0=0$. From
eqn.\eqref{kkkazzqqav} the integral condition implies $\mathsf{h}(S)\leq 2$ with equality if and only if $S$ is narrow.
\end{proof}

The following observation is crucial:
\begin{pro} \label{proe8}  $\pi\colon\ce\to \mathbb{P}^1$ a (relatively minimal) rational elliptic surface. Let $S$ be
an \emph{integral} section of $\pi$. Then the divisor
\be
\cc= S-S_0-F
\ee
is an $E_8$-root curve.
\end{pro}
\begin{proof} We have to check the three conditions in eqn.\eqref{3condd}
\be
\begin{split}
&F\cdot(S-S_0-F)=0,\qquad S_0(S-S_0-F)=S_0\cdot S=0,\\
&(S-S_0-F)^2=-2-2\, S\cdot S_0=-2.
\end{split}
\ee
So $\cc$ is an actual rational curve on the surface $\ce$ which represents in $\mathsf{NS}(\ce)$ a root of the lattice $E_8^-$ (cfr.\! eqn.\eqref{nnnssssa}).
\end{proof}

Note that to an integral section there are associated both a $(-1)$-curve $S$ and an $E_8$-root $(-2)$-curve $\cc$. If, in addition, $S$ is narrow, 
\be\label{gooodposition}
F_{u,\alpha}\cdot \cc=0\quad \text{for all }u,\;\alpha.
\ee
We say that an $E_8$-root curve  is \textit{in good position} in the N\'eron-Severi lattice if it satisfies
eqn.\eqref{gooodposition}. $E_8$-root curves in good position are in one-to-one correspondence with the integral-narrow sections of $\pi$.

\subsection{Arithmetics of SW differentials}\label{swari}
We return to the study of rank-1 special geometries and their SW differentials.

\subsubsection{The
``no dangerous irrelevant operator'' property}\label{kkkaqwhh}

Let us consider a special geometry
$X_0=\ce_0\setminus F_\infty$ described by a certain rational function $\mathscr{J}_0$ consistent with a given fiber configuration $\{F_\infty;F_i\}$. 
From eqn.\eqref{llllasqw} and the discussion following it, we see that $X_0$ carries a symplectic form $\Omega_0$ such that (in cohomology)
\be\label{kkkaqwe}
[\Omega_0]\in \Lambda_\C\equiv \Lambda\otimes \C.
\ee
Now let us slightly deform the rational function $\mathscr{J}=\mathscr{J}_0+\delta\cj$, in a way consistent with the given fiber configuration $\{F_\infty;F_i\}$, while keeping fixed the fiber at infinity 
(i.e.\! the asymptotic geometry as $u\to\infty$, see discussion around eqn.\eqref{jjjzzzh}). Since we keep fixed the UV geometry, the deformation $X_0\to X$ should correspond to a small change of  masses and relevant couplings.

The deformed manifold $X$
is smoothly equivalent to $X_0$; 
so we may identify the cohomology groups
$H^2(X,\C)\cong H^2(X_0,\C)$  and compare the symplectic forms in cohomology \cite{donagi}. The variation
$\delta[\Omega]=[\Omega]-[\Omega_0]$ computes the  modification of the masses  induced by the variation $\delta\cj$ of Kodaira's functional invariant.
Eqn.\eqref{kkkaqwe} identifies the space 
of mass parameters with a subspace of the essential vector space $\Lambda_\C$.

It is natural to require our geometry to have  ``enough'' mass deformations (or equivalently ``enough'' SW differentials) to span all $\Lambda_\C$, that is, to require that no mass deformation is forbidden or obstructed. This requirement formalizes the physical idea that we are probing all genuine IR deformations of our QFT, and not arbitrarily restricting the parameters to some special locus in coupling space. 
We call this condition \textit{SW completeness}.
The main goal of this subsection is to show the following 

\begin{cla} In rank-1, SW completeness implies the property ``no dangerous irrelevant operators'' conjectured in refs.\!{\rm\cite{Argyres:2015ffa,Argyres:2015gha,Argyres:2016xua,Argyres:2016xmc,Argyres:2016yzz}. } 
\end{cla}

\begin{proof} The statement of SW completeness says that the total number
$\boldsymbol{n}$  of deformation of an UV complete geometry  should be equal to the dimension of the space $\Lambda_\C$ plus the number of relevant/marginal operators. In formulae
\be\label{swswcon}
\boldsymbol{n}- \dim \Lambda_\C=\begin{cases} 1 & \text{if }\Delta\leq 2\\
0 &\text{otherwise.}
\end{cases}
\ee
From eqn.\eqref{rankEu}
\be\label{kkkazqg}
\dim\Lambda_\C= 8-\sum_u\mathrm{rank}\,R(F_u)=8-\sum_u e(F_u)+s+2a^\circ+2a^\ast= s+2a^\circ+2a^\ast-4.
\ee
while, from eqn.\eqref{kkaswo},  
\be
\boldsymbol{n}=s+a^\circ+2a^* -3,
\ee 
so that the \textsc{lhs} of eqn.\eqref{swswcon} is simply
\be
1-a^\circ
\ee
from which we see that
$a^\circ=1$ if $\Delta> 2$ and $a^\circ=0$ otherwise.
\end{proof}

Comparing with \S.\ref{uvcccom} we get
\begin{fact}
\textit{In a non-constant, UV and SW complete, rank-1 special geometry, an $\text{\rm additive}^\circ$ fiber (i.e.\! types $II$, $III$, and $IV$) may be present in $\ce$ only as the fiber at infinity $F_\infty$. In this case the $\cn=2$ QFT is a mass-deformation of a SCFT with $\Delta= 6,4$ and $3$, respectively.}
\end{fact}
This statement has identical implications for the classification program (in rank-1) as the ``no dangerous irrelevant operator'' conjecture of ref.\!\cite{Argyres:2015ffa,Argyres:2015gha,Argyres:2016xua,Argyres:2016xmc,Argyres:2016yzz}.

\subsubsection{The flavor lattice (elementary considerations)}\label{jjzmmmcx}

In the previous subsection we have identified $\Lambda_\C$ with the complexification $\mathfrak{h}_\C=\mathfrak{h}\otimes \C$ of the  flavor Cartan sub-algebra $\mathfrak{h}\subset\mathfrak{f}$. The dimensions of the two spaces agree for SW complete geometries.

Inside the Cartan algebra $\mathfrak{h}$ we have natural lattices, such as the weight and roots lattices of $\mathfrak{f}$. These lattices are endowed with a positive-definite symmetric pairing with respect to which the Weyl group $\mathrm{Weyl}(\mathfrak{f})$ acts by isometries. Moreover, in $\mathfrak{h}$ we may distinguish finitely many vectors playing special roles, such as
 the co-roots, the roots, and the fundamental weights.

In order for the identification $\Lambda_\C\leftrightarrow \mathfrak{h}_\C$ to be fully natural, the above discrete structures should be identifiable in $\Lambda_\C$ too. In $\Lambda_\C$ there exist canonical lattices, like $\Lambda$, $\Lambda^\vee$ and their sub- and over-lattices,
as well as a natural positive-definite symmetric pairing, i.e.\! the N\'eron-Tate height $\langle-,-\rangle_\text{NT}$. These lattices also contains a special finite sub-set, namely the integral sections.
 
In particular, to a given fiber configuration $\{F_u\}_{u\in U}$ we may associate the group $O(\mathsf{MW}(\ce)^0)$ of isometries of the narrow Mordell-Weil lattice $\mathsf{MW}(\ce)^0$.
Then, consistency yields
\begin{nec} Let $\mathfrak{f}$ be the flavor Lie algebra associated to a rank-1 (UV and SW complete) special geometry, and let $\mathrm{Weyl}(\mathfrak{f})$ be its Weyl group. Then
\be
\mathrm{Weyl}(\mathfrak{f})\subseteq O(\mathsf{MW}(\ce)^0).
\ee
\end{nec}

This condition does not fix $\mathfrak{f}$ uniquely. For instance, let  $\mathsf{MW}(\ce)^0\cong D_4$, so that
\be
O(\mathsf{MW}(\ce)^0)\cong \mathrm{Weyl}(D_4)\rtimes \mathfrak{S}_3,
\ee
where the symmetric group $\mathfrak{S}_3$ acts by Spin(8) triality. Then $O(\mathsf{MW}(\ce)^0)\cong \mathrm{Weyl}(F_4)$,
while the subgroup $\mathrm{Weyl}(D_4)\rtimes \Z/2\Z$ is isomorphic to $\mathrm{Weyl}(C_4)\cong \mathrm{Weyl}(B_4)$, so in this case the above condition leaves us with 4 possible irreducible $\mathfrak{f}$, namely $B_4$, $C_4$, $D_4$ and $F_4$, and a few  more reducible candidates.

In order to unfold the ambiguity, we need to understand the flavor root \emph{system} and not just its root \emph{lattice}. This issue will be discussed in the next subsection. The obvious guess is that the finite set of integral sections will play the major role.

In simple situations the
correct physical flavor symmetry may be easily guessed from the narrow Mordell-Weil lattice $\mathsf{MW}(\ce)^0$. However, in general, one needs the precise treatment in terms of roots systems described in the next subsection. Here we present the simplest possibile situation (i.e.\! maximal symmetry for the given $\Delta$)
where naive ideas suffice.

\begin{exe}[Maximal flavor symmetry]\label{zzzm1a} 
Let us consider fiber configurations of the form $\{F_\infty;I_1^{12-e(F_\infty)}\}$
where $F_\infty$ is one of the 7 semi-simple additive fibers in eqn.\eqref{idedeime} or $I^*_{b\leq 4}$ in the asymptotic-free case. These  configurations are the ``general deformations'' of the SCFT associated to the given fiber at infinity, in the sense that they yield the family of elliptic surfaces depending on the largest number of  parameters. Thus $\{F_\infty;I_1^{12-e(F_\infty)}\}$ is the configuration which, for a given Coulomb dimension $\Delta$ (encoded in $F_\infty$), maximizes the rank of the flavor group, see eqn.\eqref{kkkazqg}.
In this case all fibers are irreducible except (possibly) the fiber at infinity. The Mordell-Weil group is torsionless \cite{bookMW} and thus
\be
\mathsf{MS}(\ce)\cong (\mathsf{MS}(\ce)^0)^\vee\equiv \Lambda^\vee.
\ee 
Standard facts about lattices \cite{bookMW} yield

\begin{lem} Let $F_\infty$ be one of the 7 semi-simple additive fiber types in eqn.\eqref{inftytypes} or $I^*_{b\leq 2}$ and $R(F_\infty)$ the corresponding $ADE$ root system (table \ref{kodairatypes}). Let $\Lambda=R(F_\infty)^\perp$ be its orthogonal complement in the $E_8$ lattice (i.e.\! the \emph{essential} lattice). Then $\Lambda$ is an irreducible root lattice of type $ADE$,
except for $F_\infty=I^*_2$ where $\Lambda$ is the root lattice of $\mathfrak{so}(4)=A_1\oplus A_1$. ($\Lambda^\vee$ is then the corresponding $ADE$ weight lattice). See table \ref{kkkaz410}. Moreover,
\be
\mathsf{MS}(\ce)/\mathsf{MS}(\ce)^0\equiv \Lambda^\vee\big/\Lambda=R(F_\infty)^\vee\big/R(F_\infty)\equiv Z(F_\infty)\equiv Z(\ce),
\ee
is the center of the  corresponding (simply-connected)
$ADE$ Lie group.
\end{lem}

\begin{table}
\begin{center}
\begin{tabular}{c|ccccccc|cc}\hline\hline
$F_\infty$ & $II$ & $III$ & $IV$ & $I^*_0$ &
$II^*$ & $III^*$ & $IV^*$ & $I_1^*$ & $I_2^*$\\\hline
$R(F_\infty)$ & - & $A_1$ & $A_2$ & $D_4$ &
$E_8$ & $E_7$ & $E_6$\ & $D_5$ & $D_6$\\\hline
$\mathsf{MW}(\ce)^0\equiv \Lambda$ & $E_8$ & $E_7$ & $E_6$ & $D_4$ & - & $A_1$ & $A_2$ & $A_3$ & $A_1\oplus A_1$\\\hline
$\delta(F_\infty)\phantom{\Big|}$ & 0 & $\tfrac{1}{2}$ & $\tfrac{2}{3}$
& 1 & 2& $\tfrac{3}{2}$ & $\tfrac{4}{3}$ & 1, $\tfrac{5}{4}$ & 1, $\tfrac{3}{2}$\\\hline
integral repr. & - & \textbf{56} & $\mathbf{27}\oplus\mathbf{\overline{27}}$ & $\mathbf{8}_c\oplus\mathbf{8}_s\oplus \mathbf{8}_c$ & - & $\mathbf{2}$  &$\mathbf{3}\oplus\mathbf{\overline{3}}$ &$\mathbf{6}\oplus \mathbf{4}\oplus\mathbf{\overline{4}}$ & $\mathbf{4}\oplus\mathbf{2}_L\oplus\mathbf{2}_R$\\\hline\hline
\end{tabular}
\caption{Flavor symmetries and integral representations for ``general'' deformations. Note that in the non-semisimple cases $\delta(F_\infty)$ takes two distinct values (cfr.\!\cite{bookMW} page 124). For $I^*_{b\leq2}$ the integral representation is given by $\mathsf{vector}\oplus\mathsf{spinor}_L\oplus\mathsf{spinor}_R$ of $SO(8-2b)$.}\label{kkkaz410}
\end{center}
\end{table}

\begin{rem} Note that for $\{F_\infty;I_1^{12-e(F_\infty)}\}$ adding/deleting $\ast$ on the fiber at $\infty$ simply interchanges the two orthogonal sub-lattices  
$R(F_\infty)\leftrightarrow \mathsf{MW}(\ce)^0$.
\end{rem}

 In the $\{F_\infty; I_1^{12-e(F_\infty)}\}$ case, 
 for all sections $S\in\mathsf{MW}(\ce)$
\be\label{mmmmmuuasq}
\mathsf{h}(S)\equiv \langle S,S\rangle_\text{NT}=2+2\, S\cdot S_0-\begin{cases}
0 & \text{if }S\in \mathsf{MW}(\ce)^0\\
\delta(F_\infty) & \text{if }S\not\in \mathsf{MW}(\ce)^0\
\end{cases}
\ee
see table \ref{kkkaz410}.
For $F_\infty=II$, $\mathsf{MW}(\ce)^0\equiv
\mathsf{MW}(\ce)$, so the second case in \eqref{mmmmmuuasq} does not appear.

The roots of the $ADE$ lattice $\mathsf{MW}(\ce)^0$ are narrow of height $2$ hence integral sections
by \textbf{Lemma \ref{kkkz1i}} which are related to \textit{$E_8$-root curves} by \textbf{Proposition \ref{proe8}}. Being narrow, they are automatically in \emph{good position.}
It is known that the flavor Lie algebra $\mathfrak{f}\equiv \mathfrak{Lie}(\mathsf{F})$ of the ``maximally symmetric'' models is the simply-laced Lie algebra $\mathsf{MW}(\ce)^0$.
Thus the roots of the flavor algebra simply correspond to the $E_8$-root curves in good position for the fiber configuration $\{F_\infty;I_1^{12-e(F_\infty)}\}$.

The Mordell-Weil group  $\mathsf{MW}(\ce)$ is the weight lattice of the Lie algebra in the third row of table \ref{kkkaz410}, and the integral sections which are not $ADE$ roots form the weights of the representation in the last row of the table. These sections  correspond to $(-2)$-curves which are not in good position.
They form a Weyl invariant set of weights. Note that the `integral representation' of $\mathsf{F}$ in the last row of the table  is precisely the one carried by the BPS hypermultiplets which are stable in the regime $u\to\infty$.
For instance, for $\{I^*_b,I_1^{6-b}\}$, which corresponds to $SU(2)$ SQCD with $N_f=4-b$,
we get $\mathsf{F}=SO(2N_f)$ and the hypers
(quarks, monopoles, and dyons \cite{SW0,SW1}) belong to the vector and left/right spinor representations. 
\end{exe}

\begin{exe}\label{xxxzvx} In \textbf{Example \ref{zzzm1a}}
we excluded two possible fibers at $\infty$,
$I_4^*$ and $I_3^*$. The first one, which corresponds to pure SYM, has a flavor group of rank 0. The second one, i.e.\! $SU(2)$ SQCD with $N_f=1$ (cfr.\! \S.\,\ref{uvcccom}), has a flavor group of rank 1. However, in this case the flavor group is \emph{not} semi-simple, but rather the Abelian group $SO(2)$ (baryon number) which does \emph{not} correspond to a root system. Correspondingly, in this instance the essential lattice is \emph{not} a root lattice but rather \cite{bookMW}
\be
\Lambda=\langle 4\rangle,\qquad \mathsf{MW}(\ce)\equiv \Lambda^\vee=\langle 1/4\rangle,
\ee
where $\langle\ell\rangle$ stands for the group $\Z$ endowed with the quadratic form $\mathsf{h}(n)=\ell\,n^2$. One has $\delta(I^*_3)=1$,
or $\tfrac{7}{4}$, so that the integral sections correspond to the elements of $\langle 1/4\rangle$ having height $1$ or $\tfrac{1}{4}$. They correspond to $U(1)\cong SO(2)$ baryon charges $\pm1$ and $\pm\tfrac{1}{2}$, which are the correct values for quarks and, respectively, dyons in $N_f=1$ SQCD.
\end{exe}

\subsection{The flavor root system}

\subsubsection{The root system associated to the Mordell-Weil lattice}
\label{hhhazq870}

The Mordell-Weil lattices contain a canonical root system that we now define.
\smallskip 

As reviewed above, for a rational elliptic surface $\ce$ we have
\be\label{jjazqwe}
\begin{matrix}\mathsf{MW}(\ce)^0 &\subset& \mathsf{MW}(\ce)/\mathsf{MW}(\ce)_\text{tors.}&\subset &\mathsf{NS}(\ce)^-_\BQ\\
\| &&\| &&\|\\
\Lambda & \subset & \Lambda^\vee &\subset & \boldsymbol{U}_\BQ\oplus (E_8\otimes \BQ) 
\end{matrix}
\ee
$\Lambda$, $\Lambda^\vee$ being equipped with the N\'eron-Tate pairing and $\mathsf{NS}(\ce)_\BQ^-$ with \emph{minus} the intersection pairing.
The embeddings in \eqref{jjazqwe} are isometries. We consider the sublattice of
``integral points'' in $\Lambda^\vee$
\be\label{xxxxxz1f}
\Lambda_\Z:= \Lambda^\vee \cap \mathsf{NS}(\ce)^- \subset \Lambda^\vee.
\ee
A vector $s\in \Lambda_\Z$, being an element of $\mathsf{NS}(\ce)^-$, defines a divisor $D(s)$ unique up to linear equivalence. An element $\lambda\in \Lambda^\vee$ defines a section $S(\lambda)$ unique up to torsion.

The \emph{level} of $s\in \Lambda_\Z$ is the largest positive integer $k(s)$ such that 
\be
\hat s\equiv \frac{1}{k(s)}\,s\in \Lambda^\vee.
\ee
We have,
\begin{align}
\langle s,s\rangle_\text{NT}&= k(s)^2\,\mathsf{h}(\hat s) &&\forall\; s\in \Lambda_\Z\\
\langle \lambda, s\rangle_\text{NT}&= k(s)\,\langle \lambda, \hat s\rangle_\text{NS}=
- S(\lambda)\cdot D(s)\in\Z &&\forall\; s\in \Lambda_\Z,\ \lambda\in\Lambda^\vee,
\end{align}
In particular, $\Lambda_\Z\subset \Lambda$.

\begin{defn} The \textit{MW root system}
$\Xi\subset \Lambda_\Z$ is the set of elements $s\in\Lambda_\Z$ such that 
\be\label{kkzaqw}
\mathsf{h}(s)/k(s)\equiv k(s)\,\mathsf{h}(\hat s)=2\quad\Rightarrow\quad \langle s,s\rangle_\text{NT}=2\, k(s).
\ee 
\end{defn}

For each $s\in \Xi$ we consider the reflection 
\be\label{kkzasqw99}
r_s\colon \lambda\mapsto \lambda -\frac{2\langle \lambda, s\rangle_\text{NT}}{\langle s,s\rangle_\text{NT}}\,s.
\ee
\begin{lem} Let $s\in\Xi$. The reflection $r_s$:
\begin{itemize}
\item[\bf 1)]
is an isometry of $\Lambda^\vee$;
\item[\bf 2)] preserves the lattice $\Lambda_\Z$;
 \item[\bf 3)] preserves the level $k(s^\prime)$ of $s^\prime\in\Lambda_\Z$.
\end{itemize}
\end{lem}
\begin{proof} \textbf{1)} It suffice to show that $r_s(\lambda)$ is a linear combination of elements of $\Lambda^\vee$ with integral coefficients. For all $s\in \Xi$ and $\lambda\in\Lambda^\vee$,
\be
\frac{2\langle \lambda, s\rangle_\text{NT}}{\langle s,s\rangle_\text{NT}}\,s=
\frac{2\langle \lambda, s\rangle_\text{NT}}{2\, k(s)}\, k(s)\,\hat s= -S(\lambda)\cdot D(s)\,\hat s.
\ee
\textbf{2)}
We have to show that
\be
\frac{2\langle s^\prime, s\rangle_\text{NT}}{\langle s,s\rangle_\text{NT}}\in\Z\quad\text{for all }s\in \Xi,\ s^\prime\in \Lambda_\Z.
\ee
Now
\be
\frac{2\langle s^\prime, s\rangle_\text{NT}}{\langle s,s\rangle_\text{NT}}=
\frac{2k(s) \langle s^\prime, \hat s\rangle_\text{NT}}{2k(s)}=-D(s^\prime)\cdot S(\hat s)\in\Z
\ee
\textbf{3)}
Indeed,
$r_s(s^\prime)= k(s^\prime)\, r_s(\hat s^\prime)$ where $r_s(\hat s^\prime)\in \Lambda^\vee$ by \textbf{1)}.
\end{proof}

From this \textbf{Lemma} it follows
that the finite set $\Xi$ is a reduced root system canonically associated to the
Mordell-Weil group.

\subparagraph{The restricted root system of $(\ce,F_\infty)$.} In our set-up, we have a marked additive 
fiber $F_\infty\in \ce$. We consider the subset of $\Xi_\infty\subset \Xi$ such that
\be
s\in \Xi_\infty \quad\Longleftrightarrow\quad
s\in \Xi\ \text{and }S(\hat s)\ \text{crosses }F_\infty\ \text{in the identity component}.
\ee  
From \eqref{kkzasqw99} we see that 
$\Xi_\infty$ is also a root system. Indeed,
for all $s\in\Xi_\infty$, $s^\prime\in\Xi$
and $\alpha\geq 1$,
\be
F_{\infty,\alpha}\cdot S(r_s(\hat s^\prime))=
F_{\infty,\alpha}\cdot S(\hat s^\prime)-
\langle \hat s^\prime, s\rangle_\text{NT}\,F_{\infty,\alpha}\cdot S(\hat s)\equiv F_{\infty,\alpha}\cdot S(\hat s^\prime).
\ee

\subparagraph{Explicit formulae for divisors.}
Let $s\in \Lambda_\Z$ be an element of level $k(s)$, and write $\hat S$ for $S(\hat s)$. Then the $D(s)$, $S(s)$ are the divisors
\begin{align}
D(s)&= k(s)\,\Phi_\BQ(\hat S)\in \Lambda_\Z\subset\mathsf{NS}(\ce)\\
S(s)&=k(s)\big(\Phi_\BQ(\hat S)+F)+S_0\in \mathsf{NS}(\ce).
\end{align}
$S(s)$ is an exceptional (-1)-curve, i.e.\! $K_\ce\cdot S(s)=S(s)^2=-1$, namely a section.

\begin{rem} All $s\in\Lambda^\vee\cong\mathsf{MW}(\ce)/\mathsf{MW}(\ce)_\text{tor}$ corresponding to \emph{narrow-integral} sections are elements of $\Xi_\infty$ corresponding to ``short'' roots (height $=2$). Conversely, all roots of height 2 arise from narrow-integral sections. Let $\hat S$ be a non-narrow integral section which is narrow at $\infty$, and $k(\hat S)$ the smallest integer such that $k(\hat S)\, \hat S\in\Lambda$. If $\hat S$
satisfies the criterion
\be\label{criterion}
k(\hat S)\,\mathsf{h}(\hat S)=2,
\ee
then $k(\hat S)\,\hat S\in\Xi_\infty$.
\end{rem}

\begin{rem} We have $\mathrm{rank}\;\Xi_\infty\leq \mathrm{rank}\,\Lambda$.
When the inequality is strict, $\mathsf{F}$ has an Abelian factor $U(1)^a$ with $a=\mathrm{rank}\,\Lambda-\mathrm{rank}\,\Xi_\infty$, cfr.\! \textbf{Example \ref{xxxzvx}}. 
\end{rem}

\subsubsection{SW differentials and flavor}

In \S.\,\ref{jjzmmmcx} we considered the 
polar divisor of $\lambda$ up to algebraic (or linear) equivalence. In doing this we lost some information about the actual curves $\cs_i\subset\ce$ along which the SW differential $\lambda$ has poles.
We know that these curves must be sections of $\pi\colon \ce\to \mathbb{P}^1$, i.e.\! $F\cdot \cs_i=1$. We may take one of the $\cs_i$, say $\cs_0$ as the zero section $S_0\equiv\cs_0$.
The divisors dual to the free mass parameters (cfr.\! eqn.\eqref{llllasqw}) then take the form
$L_a\sim \cs_a-\cs_0$ for $a>0$.
 The $L_a$ should be trivial at infinity (since the masses are UV irrelevant), that is,
the sections $\cs_i$ should cross $F_\infty$ in the identity component\footnote{\ We call such sections \emph{narrow at $\infty$.}}, $F_{\infty,\alpha}\cdot\cs_i=\delta_{\alpha,0}$.
Eqn.\eqref{hhhas} yields
\be\label{xxkawu}
\text{\it i)}\quad L_a\equiv \cs_a-\cs_0\in \Lambda_\Z\oplus \Z F, \qquad \text{\it ii)}\quad
F_{\infty,\alpha}\cdot (L_a+\cs_0)=0\qquad \alpha\geq 1.
\ee
We have to determine the sections $\cs_i$ (equivalently, the divisors $L_a$ satisfying \textit{i)} and \textit{ii)}),
which may actually appear in the polar divisor of $\lambda$. From comparison with $E_8$ Minahan-Nemeshanski we know that $L_a$ is allowed to be an $E_8$-root $(-2)$-curve. Note that eqn.\eqref{xxkawu} enforces the condition that $L_a$ is in good position. If $L_a$ is a $E_8$-root satisfying \eqref{xxkawu},  the associated $(-1)$-curve
$\cs_a$ is an integral-narrow section hence an element of $\Xi_\infty$ of height 2. 

However the integral-narrow sections cannot be the full story,
since the set of integral-narrow sections does not behave properly under covering maps (discrete gaugings in the QFT language). In the next section we shall discuss the functorial properties of the Mordell-Weil lattices under such coverings. There it will be shown that a natural finite set of sections which contains the integral-narrow ones and behaves well under covering maps is the set $\Xi$ defined in \S.\,\ref{hhhazq870}. As we have seen, $\Xi$ is automatically a root system in $\Lambda_\R$. The condition \eqref{xxkawu} restricts further to the subsystem $\Xi_\infty$. 
Therefore consistency leaves us with just one possible conclusion:
\begin{quote}
{\it The root system of the flavor Lie group $\mathsf{F}$ is $\Xi_\infty$.}
\end{quote}

This statement is checked in \S.\ref{lllazzzq} in (essentially all) examples.
 
\begin{rem}[Abelian flavor symmetries]\label{abesym} The general situation is similar to $SU(2)$ SQCD with $N_f=1$.
In that model the rank of the flavor Lie algebra is $1$, but the set of roots is empty since: \textit{i)} by definition, in $\{I^*_3,I_1^3\}$ there are no non-narrow sections which are narrow at $\infty$, and \textit{ii)} $\Lambda=\langle 4\rangle$ so no narrow section is integral. This is the correct result for a $U(1)$ flavor symmetry. On the other hand, the integral section which are not narrow at $\infty$ give baryon numbers of BPS states as we commented in \textbf{Example \ref{xxxzvx}}. 
\end{rem}

\begin{rem}[Maximal symmetry again] In the configuration $\{F_\infty, I_1^{12-e(F_\infty)}\}$, 
all sections narrow at $\infty$ are narrow. So the roots are just the elements of $\Lambda$ which have height 2, and the root system is the unique simply-laced one with root lattice $\Lambda$, see third row of table \ref{kkkaz410}.
\end{rem}

\subsubsection{More examples of flavor root systems}\label{lllazzzq}

\begin{table}
\begin{center}
\begin{tabular}{c|cccc}\hline\hline
$\Delta$ & 6&4&3&2\\\hline
$F_\infty$ & $II$ & $III$ & $IV$ & $I_0^*$\\\hline
$R$ & $A_3$ & $A_3\oplus A_1$ & $A_3\oplus A_2$ & $A_3\oplus D_4$
\\\hline
$\mathsf{MW}(\ce)^0$ & $D_5$ & $A_3\oplus A_1$ & $(D_5:A_2)$ & $\langle 4\rangle$
\\\hline
$\mathsf{MW}(\ce)$ & $D_5^\vee$ & $A_3^\vee\oplus A_1^\vee$ & $(D_5:A_2)^\vee$
&$\langle 1/4\rangle\oplus\Z/2\Z$
\\\hline\hline
\end{tabular}
\caption{Lattices for fibers $\{F_\infty; I_4,I_1^{8-e(F_\infty)}\}$. $(D_5:A_2)$ stands for the orthogonal complement of the lattice $A_2$ in $D_5$ (it cannot be written as a direct sum of root and rank 1 lattices). The Mordell-Weil groups are read from the table attached to \textbf{Theorem 8.7} of \cite{bookMW}.}
\label{i4}
\end{center}
\end{table}

\begin{exe}[Fiber configurations $\{F_\infty; I_4,I_1^{8-e(F_\infty)}\}$] We assume the presence of a single semi-stable fiber of type $I_4$. This restricts the additive fiber at $\infty$ to 4 possible types as in table \ref{i4}.
For $F_\infty=II$ and $III$ the narrow Mordell-Weil groups are the root lattices $\mathsf{MW}(\ce)^0=D_5$ and $A_3\oplus A_1$, respectively, and the full Mordell-Weil group is the corresponding weight lattice.
The 40 roots of $D_5$ (resp.\! 14 roots of $A_3\oplus A_1$) correspond to narrow-integral sections and are roots of $\mathfrak{f}$. An integral\footnote{\ A non-integral section has height $\geq 3$.} non-narrow section $S$ has N\'eron-Tate height
\be
\mathsf{h}(S)=2-\delta(I_4)=\begin{cases} 1 & k(S)=2\\
5/4 & k(S)=4,
\end{cases}
\ee
and the criterion \eqref{criterion} is satisfied only by the sections of height 1 which have square-length 4. For $F_\infty=II$ there are 10 such roots of square-length 4, one for each vector weight in $D_5^\vee$. For $F_\infty=III$ there are 6 of them in correspondence with the vector weights of $A_3\cong\mathfrak{so}(6)$.  
We conclude:\begin{itemize}
\item The flavor Lie algebra of $\{II;I_4,I_1^6\}$ has a root system consisting of
40 roots of square-length 2 and 10 roots of square-length 4, and a Weyl group $\mathrm{Weyl}(D_5)\rtimes \Z/2\Z$. The special geometry describes a $\Delta=6$ SCFT with flavor group (isogeneous to) $Sp(10)$;
\item The flavor Lie algebra of  $\{III;I_4,I_1^6\}$ has a root system consisting of
14 roots of square-length 2 and 6 roots of square-length 4, and a Weyl group $(\mathrm{Weyl}(A_3)\rtimes \Z/2\Z)\times\mathrm{Weyl}(A_2)$. The geometry describes a $\Delta=4$ SCFT with flavor group (isogeneous to) $Sp(6)\times Sp(2)$.
\end{itemize} 
If $F_\infty=I^*_0$, $R=A_3\oplus D_4$, and $\Lambda=\langle 4\rangle$. There are no integral narrow sections, and the roots are in one-to-one correspondence with the elements of the Mordell-Weil group of N\'eron-Tate height 1.
We conclude:
\begin{itemize}
\item $\{I_0^*;I_4,I_1^2\}$ describes a $\Delta=2$ SCFT with $\mathsf{F}=Sp(2)$, namely $SU(2)$ $\cn=2^*$.
\end{itemize}

If $F_\infty=IV$, the narrow Mordell-Weil and the full Mordell-Weil groups are rank 3 dual lattices
$\Lambda$, $\Lambda^\vee$ with respective Gram matrices \cite{bookMW}
\be
\begin{bmatrix}2 & 0 &-1\\
0 & 2 &-1\\
-1 & -1& 4
 \end{bmatrix}\qquad\qquad
\frac{1}{12} \begin{bmatrix}7 & 1 &2\\
1 & 7 &2\\
2 & 2& 4
 \end{bmatrix}
\ee
Short roots are in one-to-one correspondence with elements of the first lattice of N\'eron-Tate height 2 while long roots are in correspondence with elements of the second lattice with height 1.
There are 4 short roots $[\pm 1,0,0],[0,\pm1,0]\in \Lambda$ and 4 long roots
$[\epsilon_1,\epsilon_2,-(\epsilon_1+\epsilon_2)/2]\in\Lambda^\vee$, ($\epsilon_1,\epsilon_2=\pm1$) making the root system of $Sp(4)$. 
Since the flavor group has rank 3, and there is no root associated to the last component of the Cartan subalgebra, we conclude that $\mathsf{F}$ has an Abelian factor (compare \textbf{Remark \ref{abesym}}). Then        
\begin{itemize}
\item $\{IV;I_4,I_1^4\}$ describes a $\Delta=3$ SCFT with $\mathsf{F}=Sp(4)\times U(1)$.
\end{itemize}
\end{exe}

\begin{exe}[The configuration $\{II;IV^*,I_1^2\}$]\label{exeg2}
In this case $R=E_6$ (table \ref{kodairatypes}),
 $\Lambda=A_2$ and $\mathsf{MW}(\ce)=A_2^\vee$. The 6 roots of $A_2$ yield short roots of the flavor algebra $\mathfrak{f}$.
Let us consider the roots arising from the integral non-narrow sections. An integral non-narrow section $S$ has N\'eron-Tate height (cfr.\! table \ref{kkkaz410})
\be
\mathsf{h}(S)=2-\delta(IV^*) =\frac{2}{3}\qquad \text{and }\ k(S)=3.
\ee
The criterion \eqref{criterion} is satisfied,
and  the integral non-narrow sections correspond to long roots
of square-length $2\cdot 3=6$. 
The long roots are then in 1-to-1 correspondence with the
elements of height $\tfrac{2}{3}$ in $A_2^\vee$,
whose number is $6$. The 6 roots of square-length 2 together with the 6 roots of square-length 6 form the root system of $G_2$. Therefore
\begin{itemize}
\item $\{II;IV^*,I_1^2\}$ describes a $\Delta=6$ SCFT with $\mathsf{F}=G_2$.
\end{itemize}
\end{exe}

\begin{exe}[The configuration $\{II;I^*_0,I_1^4\}$]\label{mmmmz12}
In this case $R=D_4$ (table \ref{kodairatypes}),
 $\Lambda=D_4$ and $\mathsf{MW}(\ce)=D_4^\vee$. The 24 roots of $D_4$ yield short roots of the flavor algebra $\mathfrak{f}$ of square-length 2.
The integral non-narrow sections $S$ have N\'eron-Tate height (cfr.\! table \ref{kkkaz410})
\be
\mathsf{h}(S)=2-\delta(I_0^*)=1,\qquad \text{and }\ k(S)=2,
\ee
and correspond to long roots in correspondence with the 24 elements of the lattice $D_4^\vee$ of height 1.
 We have 24 short roots of height 2 and 24 long roots of height 4; thus
 \begin{itemize}
\item $\{II;I^*_0,I_1^4\}$ describes a $\Delta=6$ SCFT with $\mathsf{F}=F_4$.
\end{itemize}   
\end{exe}

\begin{exe}[The configuration $\{II;I^*_1,I_1^3\}$]
In this case $R=D_5$ (table \ref{kodairatypes}),
 $\Lambda=A_3$ and $\mathsf{MW}(\ce)=A_3^\vee$. The 12 roots of $A_3$ yield short roots of the flavor algebra $\mathfrak{f}$ of square-length 2.
An integral non-narrow section $S$ has N\'eron-Tate height
\be
\mathsf{h}(S)=2-\delta(I_1^*)=\begin{cases} 1 & k(S)=2\\
3/4 & k(S)=4,
\end{cases}
\ee
corresponding, respectively, to the vector and spinor representations of $\mathfrak{so}(6)\cong A_3$. The first line satisfies \eqref{criterion}
and lead to 6 long roots of square-norm 4.
Then
 \begin{itemize}
\item $\{II;I_1^*,I_1^3\}$ describes a $\Delta=6$ SCFT with $\mathsf{F}=Sp(6)$.
\end{itemize}
\end{exe}

\begin{exe}[The configuration $\{II;I_4^2,I_1^2\}$]
In this case $R=A_3\oplus A_3$. This is a subtle case since two distinct 
Mordell-Weil lattices may be realized \cite{bookMW} (cfr.\! eqn.\eqref{speccase})
\begin{align}
&1) && \Lambda=A_1\oplus A_1 && \mathsf{MW}(\ce)=\Lambda^\vee\oplus \Z/2\Z\\
&2) &&\Lambda=\langle 4\rangle\oplus\langle 4\rangle&& \mathsf{MW}(\ce)=\Lambda^\vee.
\end{align}
Let us consider the two possibilities in turn.

\noindent 1) We have 4 square-length 2 roots from the integral-narrow sections. The non-narrow sections have $k(S)=2$. We have the
4 roots of square-length 4 associated to
the elements $(\pm \tfrac{1}{2},\pm\tfrac{1}{2})\in A_1^\vee\oplus A_1^\vee$. In total we get the root system of $Sp(4)$. 

\noindent 2) There are no roots from the narrow sections.
An integral section $S$ which is narrow at one of the two $I_4$ fibers has N\'eron-Tate height
\be
\mathsf{h}(S)=2-\delta(I_4)=\begin{cases} 1 & k(S)=2\\
5/4 & k(S)=4,
\end{cases}
\ee
and those which are narrow at both
\be
\mathsf{h}(S)=1-\delta(I_4)-\delta(I_4)^\prime=\begin{cases}1/4 & k(S)=4\\
1/2 & k(S)=4.
\end{cases}
\ee
The criterion \eqref{criterion} is satisfied by
the sections of N\'eron-Tate height $1$
which have square-length 4 (there are 4 of them), and by those of height $1/2$ which have square-length 8 (other 4). Rescaling the length by a factor $1/\sqrt{2}$, we recognize again the root system of
$Sp(4)$. 
Thus
 \begin{itemize}
\item $\{II;I_4^2,I_1^2\}$ describes a $\Delta=6$ SCFT with $\mathsf{F}=Sp(4)$. However, it looks like we have \emph{two} distinct theories with these properties.
\end{itemize}
\end{exe}

\begin{exe}[The configuration $\{II;I_2^\ast,I_1^2\}$] In this case we have $R=D_6$, $\Lambda=A_1\oplus A_1$ and $\Lambda^\vee=A_1^\vee\oplus A_1^\vee\cong\mathfrak{so}(4)$. The 4 roots of $\Lambda$ are roots of square-length 2 while the integral sections in the \textbf{4} of $\mathfrak{so}(4)$ give roots of square-lenght 4.
 \begin{itemize}
\item $\{II;I_2^\ast,I_1^2\}$ describes a $\Delta=6$ SCFT with $\mathsf{F}=Sp(4)$.
\end{itemize} 
\end{exe}

\begin{exe}[$\{II;I_1^*,I_3\}$] $R=D_5\oplus A_2$ and $\Lambda=\langle 12\rangle$, $\mathsf{MW}(\ce)=\langle 1/12\rangle$.
The 2 sections with $\mathsf{h}=1/3$ have
$k(S)=6$ so are roots of square-length $12=4(\sqrt{3})^2$. 
 \begin{itemize}
\item $\{II;I_1^*,I_3\}$ describes a $\Delta=6$ SCFT with $\mathsf{F}=Sp(2)$.
\end{itemize}
\end{exe}

\begin{exe}[$\{II;I_2,IV^*\}$] $R=E_6\oplus A_1$ and $\Lambda=\langle 6\rangle$, $\mathsf{MW}(\ce)=\langle 1/6\rangle$. The 2 integral section with $\mathsf{h}(S)=2/3$ have $k(S)=3$ and hence are roots of square-length
$3\cdot 2=6$. Rescaling the length, we get the root system of $\mathsf{F}=Sp(2)$, but it looks like a specialization of the $G_2$ model.
\end{exe}

\begin{exe}[$\{II;I_1,I_3^*\}$] $R=D_7$ and $\Lambda=\langle 4\rangle$, $\mathsf{MW}=\langle 1/4\rangle$. The 2 sections with $\mathsf{h}(S)=1$ have $k(S)=2$ and we get a $Sp(2)$ flavor group. 
\end{exe}

\begin{exe}[$\{II;I_4,I_2^2,I_1^2\}$] $R=A_3\oplus A_1\oplus A_1$ and $\Lambda=A_3$, $\mathsf{MW}=A_3^\vee\oplus\Z/2\Z$. We have 12 roots of length 2 from the narrow sections, and 6 roots of length 4.
 \begin{itemize}
\item $\{II;I_4,I_2^2,I_1^2\}$ describes a $\Delta=6$ SCFT with $\mathsf{F}=Sp(6)$.
\end{itemize}
\end{exe}

\begin{exe}[$\{III;I_0^*,I_1^3\}$] Here $R=A_1\oplus D_4$, $\Lambda=A_1\oplus A_1\oplus A_1$.
We have the 6 roots of the $\Lambda$
and the $3\times 4$ vectors of the three $A_1\otimes A_1$ subalgebras. The flavor algebra has 6 short and 12 long roots hence
  \begin{itemize}
\item $\{III;I_0^*,I_1^3\}$ describes a $\Delta=4$ SCFT with $\mathsf{F}=Spin(7)$.
\end{itemize}
\end{exe}

\begin{exe}[$\{III;I_1^*,I_1^2\}$] Here $R=A_1\oplus D_5$, $\Lambda=A_1\oplus \langle 4\rangle$.
We have the 2 roots of $A_1$ and the 2 sections $\mathsf{h}(S)=1$ with $k(S)$.
The two sets of roots are orthogonal\footnote{\ Of course $SU(2)\cong Sp(2)$; however we use write the two factor groups in different ways to emphasize the different role of the two symmetries in the Mordell-Weil group.}  \begin{itemize}
\item $\{III;I_1^*,I_1^2\}$ describes a $\Delta=4$ SCFT with $\mathsf{F}=SU(2)\times Sp(2)$.
\end{itemize}
\end{exe}

\begin{exe}[$\{IV;I_0^*,I_1^2\}$] $R=D_4\oplus A_2$, $\Lambda=A_2[2]$
and $\mathsf{MS}(\ce)=A_2^\vee[1/2]$.
There are no narrow-integral sections.
The integral sections which are narrow at $\infty$ correspond to the image of the $A_2$ roots in $A_2^\vee[1/2]$ which have 
$\mathsf{h}(S)=1$ and $k(S)=2$, so they are flavor roots and form a $A_2$ system.
 \begin{itemize}
\item $\{IV;I_0^*,I_1^2\}$ describes a $\Delta=3$ SCFT with $\mathsf{F}=SU(3)$.
\end{itemize}
\end{exe}

\begin{exe}[$\{II; III^*,I_1\}$]\label{kkkaqw123} In this case
$R=E_7$, $\Lambda=A_1$ and $\mathsf{MS}(\ce)=A_1^\vee$. We have two roots from the two narrow-integral sections. Non narrow integral sections have height $1/2$ and level $2$, so they do not produce any new root and
$\mathsf{F}=SU(2)$.
\end{exe}

\subsection{Classification}

The moduli space of the rational elliptic surfaces is connected; thus all geometries with a given
fiber at infinity $F_\infty$ may be obtained as degenerate limits of the ``maximally symmetric''
geometry $\{F_\infty, I_1^{12-e(F_\infty)}\}$. It is thus important to have a criterion to establish when a geometry should be considered just a special case or limit of a previous one, in which we have simply  frozen some mass deformation, and when it corresponds to a ``new'' geometry describing a different $\cn=2$ QFT.  A reasonable criterion is that we have a distinct geometry along a sub-locus $\cm^\prime\subset \cm$ in moduli space  whenever along $\cm^\prime$ there are exceptional $(-1)$-curves associated to flavor roots which are not present away from $\cm^\prime$. In other words, ``new theories''
with the same $\Delta$ correspond to loci of enhanced symmetry. 

\begin{exe} Let us consider the family of fiber configurations $\{II; I_b, I_1^{10-b}\}$,
of special geometries with $\Delta=6$.
We have
$$
\begin{tabular}{c|ccccccccc}\hline\hline
$b$ & 1 & 2 & 3  & 4 & 5 & 6 & 7 & 8 & 9\\\hline
$R$ & $-$ & $A_1$ & $A_2$ & $A_3$ & $A_4$ & $A_5$ & $A_6$ & $A_7$ & $A_8$\\\hline
$\Lambda$ &  $E_8$ & $E_7$ & $E_6$ & $D_5$ & $D_4$ & $A_3$ & $*$ & $\langle 8\rangle$ & 0\\\hline\hline
\end{tabular}
$$
A new (-1)-curve with the needed properties arises at $b=4$
where the ``unbroken'' subgroup $SO(10)\subset E_6$ get enhanced to $Sp(10)\not\subset E_6$. Then $b=5,6$ return to the subgroup symmetry and 7 and 8 to groups like $SU(2)\times U(1)$ and $U(1)$. 
\end{exe}

The evidence suggests that the above geometric criterion in terms of $(-1)$ curves produces roughly the same  restrictions as the physically motivated
``Dirac quantization constraint'' used by the authors of ref.\!\cite{Argyres:2015ffa,Argyres:2015gha,Argyres:2016xua,Argyres:2016xmc,Argyres:2016yzz}. In fact, the geometric criterion is slightly weaker than the physical one, and this aspect deserves further investigation.

The pattern emerging from the ``arithmetic'' perspective of the present paper then essentially agrees with the more direct methods of \cite{Argyres:2015ffa,Argyres:2015gha,Argyres:2016xua,Argyres:2016xmc,Argyres:2016yzz}.

\section{Base change and discrete gaugings}

In ref.\!\cite{Argyres:2016yzz} 
the non-simply-laced  flavor symmetries are understood as a result of the gauging of a discrete symmetry in a parent $\cn=2$ theory. In the arithmetic language this translates into functorial properties under base change\cite{kod1,miranda,bookMW}. In Diophantine terms, ungauging the discrete symmetry means passing from
the original special geometry (seen as an elliptic curve $E$ over the field $K=\C(u)$) to the special geometry described by the elliptic curve $E^\prime$, defined over a finite-degree extension $K^\prime$ of $K$. $E^\prime$ is given by the fibered product 
\be\label{basechange}
E^\prime:=E\otimes_K\! K^\prime.
\ee 
$K^\prime$ is the function field of some curve $C$, and the extension from $\C(u)$ to $K^\prime$ arises from a morphisms $f\colon C\to \mathbb{P}^1$. The Kodaira-N\'eron model of $E^\prime$ is an elliptic surface
$\pi\colon \ce^\prime\to C$. For our purposes we are interested in the case $C=\mathbb{P}^1$.
\smallskip

Given a rational map
$f\colon \mathbb{P}^1\to \mathbb{P}^1$
and a rational elliptic surface $\pi\colon\ce\to \mathbb{P}^1$ with section, we may pull-back the elliptic fibration through $f$ producing a new elliptic surface with section, $f^*\ce$,
not necessarily rational, on which 
the deck group of $f$ acts by automorphisms.

Suppose our relatively minimal rational elliptic surface $\ce$ has an automorphism $\alpha\colon\ce\to\ce$ which induces
the automorphism $\tau\colon\mathbb{P}^1\to\mathbb{P}^1$ on its base. If
$\mathrm{ord}(\alpha)=\mathrm{ord}(\tau)=n$, $\ce$ is the pull-back of another relatively minimal rational elliptic surface $\ce^\prime$ \emph{via}
the map
 \be
 f_n\colon z\to z^n\equiv u
 \ee 
 (we locate the fixed points of $\tau$ at $0$ and $\infty$), see \textbf{Theorem 5.1.1} of \cite{cover}. 
 \smallskip
 
In the physical terminology, $\ce^\prime$ is
 the rational elliptic surface which describes the special geometry of the QFT obtained by gauging a discrete symmetry $\Z_n$ of the parent QFT associated to $\ce$. Table (VI.4.1) of \cite{miranda} yields the change in fiber type under arbitrary local base changes. 
 Table 6 of \cite{cover} lists all possible rational elliptic surfaces which can be obtained as the pull-back of another rational elliptic surface. However not all such coverings are meaningful QFT gaugings, since, in addition, we need to impose UV and SW completeness on the geometries\footnote{ And possibly ``Dirac quantization''.}.

\subparagraph{UV and SW completeness.}
Let $f\colon z\mapsto z^n$ be a cover inducing a discrete gauging of the special geometry $\ce^{(1)}$. The functional invariants of the two geometries $\ce^{(1)}$ and $\ce^{(2)}=f^*\ce^{(1)}$ are simply related:
$\mathscr{J}^{(2)}=f^*\!\!\mathscr{J}^{(1)}$. From this relation we read the change in fiber types which affects only the fibers $F_0$ and $F_\infty$ over the branching points of $f$ in agreement with the local rules of \cite{miranda}. Semi-simplicity is preserved by base change.
Since $u$ is the Coulomb branch coordinate, UV completeness requires
\be
\Delta(F_\infty^{(1)})=\deg f\cdot\Delta(F_\infty^{(2)}).
\ee
For $\deg f>1$ we have only three possibilities $F_\infty^{(1)}=II,III,IV$.
This yields the restrictions in table \ref{restriction} which should be supplemented by the conditions arising from SW completeness. Comparing with table 5 of \cite{cover} we see that the configurations satisfying the criterion are\footnote{\ For brevity we list only the covered types which satisfy the ``Dirac quantization'' condition.}:
\begin{itemize}
\item in degree 5 none;
\item in degree 4 the single cover $\{IV^*;I_1^4\}\to\{II; III^*,I_1\}$;
\item in degree 3 the single cover $\{I_0^*; I_1^6\}\to\{II; IV^*,I_1^2\}$;
\item in degree 2 with $F_\infty^{(1)}=II$
there are seven pairs which include as covered surface the
types $\{II;I_0^*,I_1^4\}$, $\{II; I_1^*,I_1^3\}$, $\{II; I_1^*,I_3\}$, $\{II;I_2^*,I_1^2\}$, $\{II; I_3^*,I_1\}$;
\item in degree 2 with $F^{(1)}_\infty=III$
five pairs which include as covered surface the
types $\{III;I_0^*,I_1^3\}$, $\{III; I_1^*,I_1^2\}$, $\{III; I_1^*,I_2\}$, $\{II;I_2^*,I_1\}$;
\item in degree 2 with $F^{(1)}_\infty=IV$
a single cover\footnote{\ The type $\{IV;I_1^*,I_1\}$ admits a double cover of type $\{IV^2,I_2^2\}$ which does not satisfy the SW completeness criterion.} with $\{IV^*,I_1^4\}\to\{IV;I_0^*,I_1^2\}$.
\end{itemize}

\begin{table}
\begin{center}
\begin{tabular}{c|cccccccc}\hline\hline
$\deg f$ & 5 & 4 & 3 & 2& 2 & 2
\\
$F^{(1)}_\infty$ & $II$& $II$ & $II$ & $II$
& $III$ & $IV$
\\
$F^{(2)}_\infty$ & $II^*$& $IV^*$  &  $I^*_0$ & $IV$ &$I_0^*$ & $IV^*$
\\\hline\hline
\end{tabular}
\caption{Possible fibers at infinity in UV complete base changes.}\label{restriction}
\end{center}
\end{table}

For simplicity in the rest of this section we focus on the first cover in each of the above items (they are the more interesting anyhow). They have the property that the fiber $F^{(2)}_0\equiv I_0$ is smooth and hence $F^{(1)}_0\in\text{additive}^\ast\cap \text{semi-simple}$. For the five coverings we have respectively,
\be
F_\infty^{(1)}= III^*,\ IV^*,\ I_0^*,\ I_0^*,\ I_0^*.
\ee
In each case $F_0^{(1)}$ is the only reducible fiber over $u\neq\infty$.

\begin{lem}\label{zzz12uua} Let $S$ be an integral non-narrow section, narrow at $\infty$, of an elliptic surface which is the base of one of the above 5 coverings $\ce^{(2)}\to\ce^{(1)}$. One has
\be
\mathsf{h}(S)=\frac{2}{\deg f}.
\ee
Moreover $k(S)=\deg f$, except for the first degree-4 cover where $k(S)=\deg f/2$.
\end{lem}

\begin{rem} The first case corresponds to \textbf{Example \ref{kkkaqw123}} which does not present peculiarities.
\end{rem}
 
\subsection{Functoriality under base change}  Base change yields 
 a commuting diagram
\be
\begin{gathered}\xymatrix{\ce_2\ar@{-->}[rr]^{\cf} \ar[d]_{\pi_2}&& \ce_1\ar[d]^{\pi_1}\\
\mathbb{P}^1\ar[rr]_f &&\mathbb{P}^1}
\end{gathered}
\ee 
 where $\cf$ is a rational map.
 Base change \eqref{basechange} induces a map of Mordell-Weil groups
 \be
 f^\sharp\colon \mathsf{MW}(\ce_1)\to \mathsf{MW}(\ce_2).
 \ee 
 At the level of divisors $f^\sharp S$ is the closure of $\cf^*S$. The Kodaira formula yields 
 \be
 \cf^*K_{\ce^{(1)}}= \deg f\cdot K_{\ce^{(2)}}\ee

Since $S_0^{(2)}=f^\sharp S_0^{(1)}$, 
$f^\sharp$ maps integral sections into integral sections (as expected from the Number Theoretic analogy). 
One has \cite{bookMW}
 \be
 \langle f^\sharp S,f^\sharp S^\prime\rangle_\text{NT}=\deg f\cdot \langle S,S^\prime\rangle_\text{NT},
 \ee
so the pull-back of a narrow-integral section has height $2\deg f$.

Conversely, let $S\in \mathsf{MS}(\ce_1)$ be an integral
section with $\deg f\cdot \mathsf{h}(S)=2$. Its pull-back $f^\sharp S$ would be an integral section on $\ce^{(2)}$ of N\'eron-Tate height 2, that is, an integral-narrow section
associated to an $E_8$-root curve in good position.

Comparing with \textbf{Lemma \ref{zzz12uua}} we see that in these examples the root system $\Xi_\infty(\ce^{(1)})$ is composed by elements which either are associated to $E_8$-root curves in good position on $\ce^{(1)}$ or such that there is a cover under which they become associated to $E_8$-root curves in good  position. There are rare situations in which the full set of elements of $\Lambda$ whose pull-back is associated to an $E_8$-root curve is a \emph{non-reduced} root system (see \textbf{Example \ref{kkkaqw123}}). Our prescription of considering the minimal level instead of the degree of the cover reduces the root system to the correct one.

\subsection{Explicit examples}\label{exxplexe}

We conclude with a couple of explicit examples.

\begin{exe}
We consider the $\Delta=6$ QFT with the non-simply-laced flavor group
$G_2$, already discussed  in \textbf{Exercise \ref{exeg2}} from the point of view of the Mordell-Weil root system.
The Dynkin graph of $G_2$ is obtained from the one of $D_4$
by folding it, that is, by taking the quotient by the cyclic subgroup $\Z/3\Z$ of its automorphism group $\mathfrak{S}_3$, see figure \ref{kkazq}. One expects that the $G_2$ model is a $\Z/3\Z$ gauging of 
a model with $D_4\rtimes\Z/3\Z$ flavor symmetry. The special geometry of the parent QFT should be the pull-back by the cyclic cover $z\mapsto z^3$ of the $G_2$ one. Let us check this idea by explicitly constructing the two geometries. 

\begin{figure}
$$
\begin{gathered}
\xymatrix{&&\bullet\ar@{-->}@/^1.1pc/[ddr]\\
\\
\bullet\ar@{-}[uurr]\ar@{-}[ddrr]\ar@{-}[rrr]&&&\bullet\\
\\
&&\bullet\ar@{-->}@/_1.1pc/[uur]}
\end{gathered}\quad\Longrightarrow\quad
\begin{gathered}
\xymatrix{\bullet && \bullet\ar[ll]\ar@<0.3ex>[ll]\ar@<-0.3ex>[ll]}
\end{gathered}
$$
\caption{The $G_2$ Dynkin graph as a folding of the $D_4$ one.}\label{kkazq}
\end{figure}
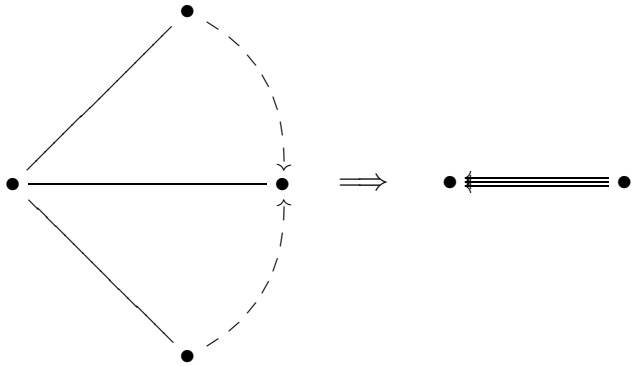

For $a,b\in\C$, let $A$ be a root of the quadratic equation
\be
A^2+2(a+b)A+(a-b)^2=0,
\ee
and set $c=(A+a+b)/2=\pm\sqrt{ab}$.
Consider the two rational functions
\begin{gather}
\mathscr{J}_1(z)= A\, \frac{z}{(z-a)(z-b)}=1-
\frac{(z-c)^2}{(z-a)(z-b)}\\
\mathscr{J}_2(w)= A\, \frac{w^3}{(w^3-a)(w^3-b)}=\mathscr{J}_1(w^3).
\end{gather}
Clearly, they are related by the base change $w\to z= w^3$ branched over $w=0,\infty$.
\smallskip

The function $\mathscr{J}_2(w)$ describes a rational elliptic surface of type $\{I_0^*;I_1^6\}$ with the fiber at infinity of type $I_0^*$ such that 
$\mathscr{J}(\infty)=\mathscr{J}(0)=0$ while the pole form two orbits under the $\Z/3\Z$ group $w\mapsto e^{2\pi i/3}w$.
Therefore, $\mathscr{J}_2(w)$ describes a very special point in the moduli space of
 of $SU(2)$ SQCD with $N_f=4$ where $\tau=e^{2\pi i/3}$ and the hyper masses are invariant under a $\Z/3\Z$ symmetry. $w$ is a global coordinate on the  $SU(2)$ Coulomb branch and has dimension $\Delta=2$. The monodromy at infinity
corresponds to $w\mapsto e^{2\pi i}w$, and is $m(I_0^*)\equiv -1\in SL(2,\Z)$.

The function $\mathscr{J}_1(z)$ has two zeros of order 1 and two simple poles. It  describes a rational elliptic surface of type $\{II; IV^*,I_1^2\}$; the $\text{additive}^\circ$ fiber $II$ should be at infinity (cfr.\! \eqref{kkkaqwhh}), so this function describes a special geometry with $\Delta=6$. 
This is obvious, since $z\equiv w^3$ has dimension $3\cdot 2=6$ while
$w\to e^{2\pi i}w$ is equivalent to $z\to e^{6\pi i}z$, that is, the two monodromies at infinity are related as
$M_2=M_1^3$, which corresponds to the identity $M(II)^3=M(I_0^*)$. The fiber at the second branch point of the cover, zero, is $IV^*$ and again $M(IV^*)^3=M(I_0)=1$.

Since the covering theory has $SO(8)\rtimes \Z/3\Z$ symmetry and the 
covered one a $G_2$ flavor symmetry and the deck group is $\Z/3\Z$, this geometry precisely corresponds to the diagram folding of figure \ref{kkazq}.

We note that
\be
Z(\ce_1)/Z(\ce_1)_\infty=\Z/3\Z,\qquad Z(\ce_2)/Z(\ce_2)_\infty=0.
\ee
In the $G_2$ geometry the group $\Z/3\Z$ acts on the sub-group of sections narrow at infinity,
the trivial representation corresponding to the subgroup of narrow sections.
$$f^*\colon\mathsf{MW}(\ce_1)\to \mathsf{MW}(\ce_2)$$
\end{exe}

\begin{exe}
We consider the rational elliptic surface of type
$\{II;I_0^*,I_1^4\}$ which describes a (mass deformed) $\Delta=6$ SCFT with $\mathsf{F}=F_4$. Its functional invariant has the form
\be
\mathscr{J}_1(z)=A\,\frac{(z-b)^3}{\prod_i(z-a_i)}.
\ee
Writing $z=w^2$ we get on the double cover a function $\mathscr{J}_2(w)$
corresponding to a surface of fiber type $\{IV;I_1^8\}$, that is, the
$\Delta=3$ model with $\mathsf{F}=E_6$ at a certain $\Z/2\Z$ symmetric point. The corresponding diagram folding is represented in figure
\ref{E6F4}. \end{exe}

\begin{figure}
\begin{equation*}
\begin{gathered}
\xymatrix{&\bullet\ar@{-->}@/^2pc/[ddrr]\\
&\bullet\ar@{-}[u]\ar@{-->}@/^1pc/[dr]\\
\bullet\ar@{-}[r]&\bullet\ar@{-}[r]\ar@{-}[u]&\bullet\ar@{-}[r]&\bullet}\end{gathered}\qquad\Longrightarrow\qquad\begin{gathered}
\xymatrix{\bullet\ar@{-}[r]& \bullet &\bullet\ar@{=>}[l]\ar@{-}[r]&\bullet}\end{gathered}
\end{equation*}
\caption{Diagram folding $E_6\to F_4$.}\label{E6F4}
\end{figure}

\end{document}